\documentclass{amsart}

\usepackage{amsmath}
\usepackage{latexsym}
\usepackage{amssymb}
\usepackage{amscd}
\usepackage{amsfonts}
\usepackage{mathrsfs}

\newcommand{\oten}{\mathop{\overline{\otimes}}}

\begin{document}

\newtheorem{thm}{Theorem}
\newtheorem{cond}{Condition}
\newtheorem{aspt}{Assumption}
\newtheorem{prop}{Proposition}
\newtheorem{lem}{Lemma}
\newtheorem{cor}{Corollary}
\theoremstyle{definition}
\newtheorem{definition}{Definition}
\newtheorem{exmp}{Example}
\newtheorem{rem}{Remark}

\title[Limit theorems for quantum stochastic models]{Approximation and limit theorems for quantum stochastic models
with unbounded coefficients}
\thanks{L.B.\ and R.v.H.\ are supported by the ARO under 
Grant No.\ W911NF-06-1-0378. A.S.\ acknowledges support from the ONR under 
Grant No.\ N00014-05-1-0420.}
\author{Luc Bouten}
\author{Ramon van Handel}
\author{Andrew Silberfarb}
\address{Physical Measurement and Control 266-33, California Institute 
of Technology, Pasadena, CA 91125, USA}
\email{bouten@its.caltech.edu,ramon@its.caltech.edu,drewsil@its.caltech.edu}

\begin{abstract}
We prove a limit theorem for quantum stochastic differential equations 
with unbounded coefficients which extends the Trotter-Kato theorem for 
contraction semigroups.  From this theorem, general results on the 
convergence of approximations and singular perturbations are obtained.  
The results are illustrated in several examples of physical interest.
\end{abstract}


\maketitle

\section{Introduction}\label{sec introduction}

It has been well established that the quantum stochastic differential
equations of Hudson and Parthasarathy \cite{HuP84} provide an extremely
rich source of realistic physical models, particularly in the field of
quantum optics \cite{Bar03}.  The physical models in this framework are
always Markovian, just like their classical counterparts, and it is the
Markov property which makes these models particularly tractable.  Several
authors have developed mathematical methodology which can be used to 
obtain these Markovian models directly from fundamental models in quantum 
field theory by applying a suitable limiting procedure 
\cite{AFLu90,Gou05,DDR06}, which places such models on a firm physical and 
mathematical foundation.  On the other hand, the simplification provided 
by the Markov property is essential in many studies, particularly in 
investigations where nontrivial dynamics plays an important role.

Despite the significant simplification provided by the Markov limit, many
realistic physical models remain rather complex, and often further
simplifications are sought.  In the physics literature it has been known
for a long time that models with multiple time scales can be significantly
simplified, in parameter regimes where these time scales are well
separated, by eliminating the fast variables from the description.  Such
singular perturbation methods are known as `adiabatic elimination' in the
physics literature, and are extremely widespread on a heuristic level.
Until recently, however, no rigoruous development was available in the
quantum probability literature.  In two recent papers \cite{GvH07,BoS07}
initial steps were made in this direction, but, as is discussed below,
these results are not entirely satisfactory.

The motivation for the current paper stems from the goal to develop a
widely applicable singular perturbation theory for quantum stochastic
differential equations.  As in previous results on the Markov limit, the
simplified model is obtained by taking the limit of a sequence of quantum
stochastic differential equations which depend on a parameter.  Our main
result, theorem \ref{thm:qsdetk} below, provides a general method to prove
the convergence of a sequence of quantum stochastic differential equations
to a limiting equation.  This technique hinges crucially on the Markov
property of all the equations involved, and extends the Trotter-Kato
theorem in the theory of one-parameter semigroups \cite{Dav80}.  Using
this technique, we develop a general singular perturbation method for
quantum stochastic differential equations, theorem \ref{thm:sing} below,
which contains and extends all known results on this topic.

\subsection*{Relation to previous work}

To our knowledge, the first rigorous result on singular perturbation for
quantum stochastic differential equations was given in the paper
\cite{GvH07}.  This paper considers an atomic system inside an optical
cavity, where the cavity is strongly coupled to an external field. It is
shown that in the strong coupling limit, the cavity can be entirely
eliminated from the model leaving an effective interaction between the
atoms and the external field. The convergence to the limiting equation is
given in the weak topology, and the proofs utilise Dyson expansion
techniques similar to those used in the literature on Markov limits.  
However, the results in this paper are limited to the particular model
under consideration, and it is unclear how to generalize these techniques
to solve general singular perturbation problems.

In \cite{BoS07} singular perturbations are studied in a general setting.
The key observation in this paper is that, due to the Markovian nature of
all the equations, the Trotter-Kato theorem can be used to prove strong
convergence.  This is done by defining certain skew semigroups on the
algebra of bounded operators on the initial space.  Though results of a
general nature are obtained in this fashion, the technique is
essentially restricted to the case that all quantum stochastic
differential equations of interest have bounded coefficients.  In the
unbounded case the semigroups on the algebra of bounded operators are
typically only weak$^*$ continuous and not strongly continuous, which
precludes the application of the Trotter-Kato theorem. Unfortunately, the
boundedness assumption excludes a large number of singular perturbation
problems of practical interest, including the problem studied in
\cite{GvH07}.

The singular perturbation results in this paper extend and bridge the gap
between the results of \cite{GvH07,BoS07}.  We consider a family of
quantum stochastic differential equations whose coefficients are not
necessarily bounded but possess a common invariant domain.  Rather than
considering semigroups on the algebra of bounded operators, we work
directly with semigroups on the initial Hilbert space.  This approach has
many far reaching consequences.  Unbounded coefficients are easily treated
in this setting, and the type of convergence which we obtain (strong
convergence uniformly on compact time intervals) is much stronger than in
\cite{GvH07,BoS07}.  Moreover, the proofs in this setting are
significantly simplified as compared to \cite{GvH07,BoS07}, which
highlights that the current approach is very natural.  In order to prove
our singular perturbation results, we first prove a general theorem on the
convergence of quantum stochastic differential equations with unbounded
coefficients, which extends the Trotter-Kato theorem to our setting.  
With this theorem in hand, the remaining work on singular perturbations is
chiefly algebraic (rather than analytic) in nature.

Though motivated by a rather different set of problems, our main
convergence result, theorem \ref{thm:qsdetk}, is closely related to
results obtained for the purpose of proving existence and uniqueness of
solutions of quantum stochastic differential equations with unbounded
coefficients \cite[proposition 3.4]{Fag93} and \cite[theorem 1.7]{LiW07},
\cite{LiW06}. Instead of showing the existence of a contractive cocycle
satisfying the limit equation, we assume the existence of a unitary
cocycle satisfying the limit equation, which allows us to obtain a much
stronger form of convergence than is obtained in \cite{Fag93,LiW07,LiW06}.

\subsection*{Organization of the paper}

The remainder of this paper is organized as follows. In section
\ref{sec:main} we introduce the class of models that we will consider and
the technical conditions which are assumed to be in force throughout the
paper.  In subsection \ref{sec:tk} we describe our main convergence
result, and subsection \ref{sec:sg} describes our general result on
singular perturbations.  Section \ref{sec:ex} is devoted to the discussion
of several examples of physical interest which demonstrate the flexibility
of our results.  Finally, section \ref{sec:proofmain} presents the proofs
of the results described in subsection \ref{sec:tk}, while section
\ref{sec:singpf} presents the proofs of the results described in section
\ref{sec:sg}.

\section{Main results}
\label{sec:main}

Throughout this paper we let $\mathcal{H}$, the initial space, be a
separable (complex) Hilbert space. We denote by
$\mathcal{F}=\Gamma_s(L^2(\mathbb{R}_+;\mathbb{C}^n))$ the symmetric Fock
space with multiplicity $n\in\mathbb{N}$ (i.e., the one-particle space is
$\mathbb{C}^n\otimes L^2(\mathbb{R}_+) \cong
L^2(\mathbb{R}_+;\mathbb{C}^n)$), and by $e(f)$, $f\in
L^2(\mathbb{R}_+;\mathbb{C}^n)$ the exponential vectors in $\mathcal{F}$.  
The annhiliation, creation and gauge processes on $\mathcal{F}$, as well
as their ampliations to $\mathcal{H}\otimes\mathcal{F}$, will be denoted
as $A_t^i$, $A_t^{i\dag}$ and $\Lambda_t^{ij}$, respectively (the channel
indices are relative to the canonical basis of $\mathbb{C}^n$).  
Moreover, we will fix once and for all a dense domain
$\mathcal{D}\subset\mathcal{H}$ and a dense domain of exponential vectors
$\mathcal{E}=\mathrm{span}\{e(f):f\in \mathfrak{S}\}\subset \mathcal{F}$,
where $\mathfrak{S}\subset L^2(\mathbb{R}_+;\mathbb{C}^n)\cap
L^\infty_{\rm loc}(\mathbb{R}_+;\mathbb{C}^n)$ is an admissible subspace
in the sense of Hudson-Parthasarathy \cite{HuP84} which is presumed to
contain at least all simple functions.  For a detailed description of
these definitions and of the Hudson-Parthasarathy stochastic calculus 
which we will use throughout this paper, we refer to \cite{HuP84,Par92,Bar03}.

For every $k\in\mathbb{N}$ we consider a quantum stochastic 
differential equation of the form
\begin{equation}
\label{eq:hpprelim}
	dU^{(k)}_t = 
	U^{(k)}_t\left\{
	\sum_{i,j=1}^n (N_{ij}^{(k)}-\delta_{ij})\, d\Lambda^{ij}_t 
	+ \sum_{i=1}^n M_i^{(k)} dA^{i\dag}_{t}
	+ \sum_{i=1}^n L_i^{(k)} dA^i_t  + K^{(k)} dt
	\right\}, 
\end{equation} 
where $U^{(k)}_0 =I$ and the quantum stochastic integrals are defined
relative to the domain $\mathcal{D}\oten\mathcal{E}$.  The purpose of this 
paper is to prove that, when the dependence of the coefficients on $k$ is 
chosen appropriately, the solutions $U_t^{(k)}$ converge as $k\to\infty$ 
in a suitable sense to the solution of a limit quantum stochastic 
differential equation
\begin{equation}
\label{eq:hplim}
	dU_t = 
	U_t\left\{
	\sum_{i,j=1}^n (N_{ij}-\delta_{ij})\, d\Lambda^{ij}_t 
	+ \sum_{i=1}^n M_i\, dA^{i\dag}_{t}
	+ \sum_{i=1}^n L_i\, dA^i_t  + K\, dt
	\right\}.
\end{equation}
Of course, the form of $K$, etc.\ will depend on our choice for $K^{(k)}$, 
etc., and a part of our task will be to identify the appropriate limit 
coefficients.  We will both prove a general convergence result---a version 
of the Trotter-Kato theorem for quantum stochastic differential 
equations---and investigate in detail a large class of singular 
perturbation problems (known as \textit{adiabatic elimination} in the 
physics literature) which are of significant practical importance in 
obtaining tractable simplifications of complex physical models.

As it turns out, it is not always the case that $U_t^{(k)}$ converges on
the entire Hilbert space $\mathcal{H}\otimes\mathcal{F}$; in singular
perturbation problems the limit will only exist on a closed subspace
$\mathcal{H}_0\otimes\mathcal{F}$, while for vectors in
$\mathcal{H}_0^\perp\otimes\mathcal{F}$ the limit becomes increasingly
singular as $k\to\infty$.  We will return to this phenomenon in detail 
later on; for the time being, suffice it to say that we must take into 
account the possibility that the limiting equation (\ref{eq:hplim}) lives 
on a smaller space $\mathcal{H}_0\otimes\mathcal{F}$ than the prelimit 
equations (\ref{eq:hpprelim}).  For the time being, we will fix the 
closed subspace $\mathcal{H}_0\subset\mathcal{H}$ and the dense domain
$\mathcal{D}_0\subset\mathcal{D}\cap\mathcal{H}_0$ in $\mathcal{H}_0$.
The quantum stochastic integrals in (\ref{eq:hplim}) will then be defined 
relative to the domain $\mathcal{D}_0\oten\mathcal{E}$.

We begin by imposing conditions on these equations that are assumed to be
in force throughout this paper.  For notational simplicity, we use the
same notation for operators on $\mathcal{H}$ or on $\mathcal{F}$ and for
their ampliations to $\mathcal{H}\otimes\mathcal{F}$ (and similarly for
$\mathcal{H}_0$ and $\mathcal{H}_0\otimes\mathcal{F}$, etc.)  Throughout,
we denote by $^\dag$ an adjoint relationship on the domain of definition,
i.e., two operators $X,X^\dag$ on the domain $\mathcal{D}$ (or
$\mathcal{D}_0$) are always presumed to satisfy the adjoint pair property
$\langle u,Xv\rangle = \langle X^\dag u,v\rangle$ for all
$u,v\in\mathcal{D}$ (resp.\ $\mathcal{D}_0$).

\begin{cond}[Hudson-Parthasarathy conditions]
\label{cond:invariance}
For all $k\in\mathbb{N}$ and $1\le i,j\le n$, the operators $K^{(k)}$,
$K^{(k)\dag}$, $L_i^{(k)}$, $L_i^{(k)\dag}$, $M_i^{(k)}$, $M_i^{(k)\dag}$,
$N_{ij}^{(k)}$, $N_{ij}^{(k)\dag}$ have the common invariant
domain $\mathcal{D}$, and the operators $K$, $K^\dag$, $L_i$,
$L_i^\dag$, $M_i$, $M_i^\dag$, $N_{ij}$, $N_{ij}^\dag$
have the common invariant domain $\mathcal{D}_0$.
In addition, we assume that the Hudson-Parthasarathy conditions
\begin{eqnarray*}
	&&K+K^\dag = -\sum_{i=1}^n L_iL_i^\dag,\qquad
	M_i = -\sum_{j=1}^n N_{ij}L_j^\dag, \\
	&&\sum_{j=1}^n N_{mj}N_{\ell j}^\dag =
	\sum_{j=1}^n N_{jm}^\dag N_{j\ell} =
	\delta_{m\ell}
\end{eqnarray*}
are satisfied on $\mathcal{D}_0$, and that
\begin{eqnarray*}
	&&K^{(k)}+K^{(k)\dag} = -\sum_{i=1}^n L_i^{(k)}L_i^{(k)\dag},\qquad
	M_i^{(k)} = -\sum_{j=1}^n N_{ij}^{(k)}L_j^{(k)\dag}, \\
	&& \sum_{j=1}^n N_{mj}^{(k)}N_{\ell j}^{(k)\dag} =
	\sum_{j=1}^n N_{jm}^{(k)\dag} N_{j\ell}^{(k)} =
	\delta_{m\ell}
\end{eqnarray*}
are satisfied on $\mathcal{D}$ for all $k\in\mathbb{N}$.
\end{cond}

Denote by $\theta_t:L^2([t,\infty[\mbox{};\mathbb{C}^n)\to
L^2(\mathbb{R}_+;\mathbb{C}^n)$ the canonical shift $\theta_tf(s)=f(t+s)$, 
and by $\Theta_t:\mathcal{F}_{[t}\to\mathcal{F}$ its second quantization
(here $\mathcal{F}\cong\mathcal{F}_{t]}\otimes\mathcal{F}_{[t}$ denotes 
the usual continuous tensor product decomposition). Recall that an 
adapted process $\{U_t:t\ge 0\}$ on $\mathcal{H}\otimes{\mathcal{F}}$ is 
called a \textit{contraction cocycle} if $U_t$ is a contraction for all 
$t\ge 0$, $t\mapsto U_t$ is strongly continuous and 
$U_{s+t}=U_s(I\otimes\Theta_s^*U_t\Theta_s)$ with 
$I\otimes\Theta_s^*U_t\Theta_s \in \mathcal{F}_{s]}\otimes
\mathcal{H}\otimes\mathcal{F}_{[s}\cong\mathcal{H}\otimes\mathcal{F}$.
If $U_t$, $t\ge 0$ are even unitary, the process is called a 
\textit{unitary cocycle}.

\begin{cond}[Cocycle solutions]
\label{cond:unitarity}
For all $k\in\mathbb{N}$, equation (\ref{eq:hpprelim}) possesses a unique 
solution $\{U_t^{(k)}:t\ge 0\}$ which extends to a contraction cocycle on 
$\mathcal{H}\otimes{\mathcal{F}}$, while equation (\ref{eq:hplim})
possesses a unique solution $\{U_t:t\ge 0\}$ which extends to a unitary 
cocycle on $\mathcal{H}_0\otimes\mathcal{F}$.
\end{cond}

The following elementary result is proved in section \ref{sec:proofmain}.

\begin{lem}
\label{lem:gens}
For $\alpha,\beta\in\mathbb{C}^n$, define $T_t^{(\alpha\beta)}:
\mathcal{H}_0\to\mathcal{H}_0$ such that
\begin{equation*}
	\langle u,T_t^{(\alpha\beta)}v\rangle = e^{-(|\alpha|^2+
		|\beta|^2)t/2}
	\langle u\otimes e(\alpha I_{[0,t]}),
	U_t\,v\otimes e(\beta I_{[0,t]})\rangle
	\qquad \forall\,u,v\in\mathcal{H}_0,~t\ge 0.
\end{equation*}
Then $T_t^{(\alpha\beta)}$ is a strongly continuous contraction 
semigroup on $\mathcal{H}_0$, and the generator 
$\mathscr{L}^{(\alpha\beta)}$ of this semigroup satisfies
$\mathrm{Dom}(\mathscr{L}^{(\alpha\beta)})\supset\mathcal{D}_0$ 
such that for $u\in\mathcal{D}_0$
\begin{equation}
\label{eq:geninvdomain}
	\mathscr{L}^{(\alpha\beta)}u =
	\left(\sum_{i,j=1}^n \alpha_i^* N_{ij} \beta_j +
	\sum_{i=1}^n \alpha^*_i M_i +
	\sum_{i=1}^n L_i \beta_i + K
	- \frac{|\alpha|^2+|\beta|^2}{2}
	\right)u.
\end{equation}
The same result holds for 
$T_t^{(k;\alpha\beta)}:\mathcal{H}\to\mathcal{H}$ and 
$\mathscr{L}^{(k;\alpha\beta)}$, defined by replacing $U_t$ by 
$U_t^{(k)}$ and making the obvious modifications.  In particular,
$\mathrm{Dom}(\mathscr{L}^{(k;\alpha\beta)})\supset\mathcal{D}$.
\end{lem}

We now impose one further condition: we assume that
(\ref{eq:geninvdomain}) completely determines
$\mathscr{L}^{(\alpha\beta)}$.

\begin{cond}[Core]
\label{cond:core}
For any $\alpha,\beta\in\mathbb{C}^n$, the domain
$\mathcal{D}_0$ is a core for $\mathscr{L}^{(\alpha\beta)}$.
\end{cond}

\begin{rem}
Note that it is not necessary for our purposes to assume that $U_t^{(k)}$ 
is unitary or that $\mathcal{D}$ is a core for 
$\mathscr{L}^{(k;\alpha\beta)}$, $k\in\mathbb{N}$.  Typically this will 
nonetheless be the case in physical applications; in particular, recall 
that when the coefficients are bounded, condition \ref{cond:invariance} 
guarantees that the solution of the associated quantum stochastic 
differential equation is unitary and the core condition is trivially 
satisfied.  In the unbounded case, however, the solution may only be a 
contraction due to the possibility of explosion.
\end{rem}

\begin{rem}
We have chosen the \textit{right} Hudson-Parthasarathy equations
(\ref{eq:hpprelim}) and (\ref{eq:hplim}), rather than the more familiar 
\textit{left} equations where the coefficients are placed to the left of 
the solution.  This means that the Schr\"odinger evolution of a state 
vector $\psi\in\mathcal{H}_0\otimes\mathcal{F}$ is given by $U_t^*\psi$, 
etc.  The main reason for this choice is that for quantum stochastic 
differential equations with unbounded coefficients, it is generally much 
easier to prove the existence of a unique cocycle solution for the right 
equation than for the left equation (see \cite{Fag90,LiW06}); this is 
chiefly because it is not clear that the solution should leave
$\mathcal{D}_0\oten\mathcal{E}$ invariant, so that the left equation may 
not be well defined (this is not a problem for the right equation, as the 
solution only appears to the left of all unbounded coefficients and is 
presumed to be bounded).  This appears to be an artefact of the fact that 
quantum stochastic integrals are defined on a fixed domain, and is not 
necessarily a physical problem.  

We will ultimately prove that $U_t^{(k)*}$ converges to $U_t^*$ strongly
on $\mathcal{H}_0\otimes\mathcal{F}$ uniformly on compact time intervals, 
i.e., we will show that
$$
	\lim_{k\to\infty}\sup_{0\le t\le T}\|U_t^{(k)*}\psi-
		U_t^*\psi\|= 0\qquad
	\forall\,\psi\in\mathcal{H}_0\otimes\mathcal{F},~
	T<\infty.
$$
Therefore, there is no loss of generality in working with the more 
tractable right equations.  If we wish to begin with a well defined left 
equation, our results can be immediately applied to the 
Hudson-Parthasarathy equation for its adjoint.
\end{rem}

\begin{rem}
\label{rem:fagnola}
In practice, the result of Fagnola \cite{Fag90,FRS94} provides a 
convenient sufficient condition which can be used to verify conditions 
\ref{cond:unitarity} and \ref{cond:core}.  Let us recall a slightly 
stronger version of this result (for simplicity phrased in terms of $U_t$, 
the analogous result holds for $U_t^{(k)}$).  Suppose that for all 
$u\in\mathcal{D}_0$ and $\ell\in\mathbb{N}$, there exists a constant 
$c(\ell,u)$ such that:
\begin{enumerate}
\item For all $u\in\mathcal{D}_0$ and for some $\varepsilon>0$
independent of $u$
$$
	\sum_{\ell=1}^\infty c(\ell,u)\,\varepsilon^\ell<\infty;
$$
\item For all $\ell\in\mathbb{N}$ and all choices of 
$X(1),\ldots,X(\ell)$, where each $X(\cdot)$ is one of $K$, $K^\dag$, 
$L_i$, $L_i^\dag$, $M_i$, $M_i^\dag$, $N_{ij}$, $N_{ij}^\dag$, we have for 
all $u\in\mathcal{D}_0$
$$
	\|X(1)\cdots X(\ell)u\| \le c(\ell,u)\,\sqrt{(\ell+m)!},
$$
where $m$ is the number of occurences of $K$ or $K^\dag$ in the
sequence $X(1),\ldots,X(\ell)$.
\end{enumerate}
Then equation (\ref{eq:hplim}) possesses a unique solution $\{U_t:t\ge 
0\}$ that extends to a unitary cocycle on 
$\mathcal{H}_0\otimes\mathcal{F}$, which 
verifies condition \ref{cond:unitarity}.  To show that $\mathcal{D}_0$ is 
a core for $\mathscr{L}^{(\alpha\beta)}$, it suffices to note that the 
Fagnola conditions imply that every $u\in\mathcal{D}_0$ is analytic 
\cite[definition 3.1.17]{BrR87}, so that the result follows from 
\cite[corollary 3.1.20]{BrR87}.  This verifies condition \ref{cond:core}.

If we are only interested in the existence and uniqueness of a contraction 
cocycle $U_t$, simple sufficient conditions exist in a much more general 
setting than the one considered in \cite{Fag90,FRS94}.  In particular, 
from \cite[theorem 1.1]{LiW06}, one can see that it suffices to prove that 
the closure of $\mathscr{L}^{(\alpha\beta)}$ on $\mathcal{D}_0$, as 
defined in equation (\ref{eq:geninvdomain}), is the generator of a 
strongly continuous contraction semigroup on $\mathcal{H}_0$ for every 
$\alpha,\beta\in\mathbb{C}^n$.
\end{rem}

\subsection{A Trotter-Kato theorem for quantum stochastic differential 
equations}
\label{sec:tk}

In order to prove that $U_t^{(k)}$ converges to $U_t$ as $k\to\infty$, we
would like to use an argument which requires only to verify conditions 
on the coefficients of the associated equations.  After all, the
dependence of the coefficients on $k$ is known in advance, while the
solutions of the equations could potentially depend on $k$ in a
complicated manner.  In order to set the stage for our main convergence 
result, let us shift our attention for the moment to the semigroups
$T_t^{(k;\alpha\beta)}$ and $T_t^{(\alpha\beta)}$.  The convergence of the 
former to the latter can be verified using a standard technique: the 
Trotter-Kato theorem.  For future reference, we state this result here in 
a particularly convenient form \cite[theorem 3.17, p.\ 80]{Dav80}.

\begin{thm}[Trotter-Kato]
\label{thm Trotter-Kato}
Let $\mathcal{H}$ be a Hilbert space and let 
$\mathcal{H}_0\subset\mathcal{H}$ be a closed 
subspace. For each $k\in\mathbb{N}$, let $T_t^{(k)}$ be a strongly 
continuous contraction semigroup on $\mathcal{H}$ with generator 
$\mathscr{L}^{(k)}$. Moreover, let $T_t$ be a strongly continuous 
contraction semigroup on $\mathcal{H}_0$ with generator $\mathscr{L}$. 
Let $\mathcal{D}_0$ be a core for $\mathscr{L}$.
The following conditions are equivalent:
\begin{enumerate}
\item For all $\psi \in \mathcal{D}_0$ there exist 
$\psi^{(k)} \in \mathrm{Dom}(\mathscr{L}^{(k)})$ such that 
\begin{equation*}
    \psi^{(k)} \xrightarrow{k\to\infty} \psi, \qquad 
    \mathscr{L}^{(k)}\psi^{(k)} \xrightarrow{k\to\infty}
    \mathscr{L}\psi.
\end{equation*}
\item For all $T<\infty$ and $\psi\in\mathcal{H}_0$ 
\begin{equation*}
    \lim_{k \to \infty} \sup_{0 \le t \le T} 
    \|T_t^{(k)}\psi - T_t\psi\| = 0.
\end{equation*}   
\end{enumerate}
\end{thm}

It should be noted, in particular, that the condition (a) depends directly
on the generators $\mathscr{L}^{(k)}$ and $\mathscr{L}$, whose form (in
terms of $k$) is typically known explicitly.  Condition (b), however,
entails the convergence of the semigroups (in a particularly strong form).

The central result of this paper is an extension of the Trotter-Kato
theorem to the quantum stochastic differential equations (QSDEs) defined
in equations (\ref{eq:hpprelim}) and (\ref{eq:hplim}).  The proof of the
following theorem can be found in section \ref{sec:proofmain}.

\begin{thm}[QSDE Trotter-Kato]
\label{thm:qsdetk}
The following conditions are equivalent.
\begin{enumerate}
\item
For every $\alpha,\beta\in\mathbb{C}^n$ and 
$u\in\mathcal{D}_0$, there exist
$u^{(k)}\in\mathrm{Dom}(\mathscr{L}^{(k;\alpha\beta)})$
such that
$$
	u^{(k)}\xrightarrow{k\to\infty}u,\qquad
	\mathscr{L}^{(k;\alpha\beta)}u^{(k)}\xrightarrow{k\to\infty}
	\mathscr{L}^{(\alpha\beta)}u.
$$
\item
For every $\psi\in\mathcal{H}_0\otimes\mathcal{F}$ and
$T<\infty$
\begin{equation*}
    \lim_{k \to \infty} \sup_{0 \le t \le T} 
    \|U_t^{(k)*}\psi - U_t^*\psi\| 
	= 0.
\end{equation*}   
\end{enumerate}
Moreover, if $\mathcal{D}$ is a core for all 
$\mathscr{L}^{(k;\alpha\beta)}$, $k\in\mathbb{N}$ then we can always 
choose $\{u^{(k)}\}\subset\mathcal{D}$.
\end{thm}

The convergence of $T_t^{(k;\alpha\beta)}$ to $T_t^{(\alpha\beta)}$ for
all $\alpha,\beta\in\mathbb{C}^n$ thus turns out to be equivalent to the
convergence, in a very strong sense, of $U_t^{(k)}$ to $U_t$. This is a
powerful technique precisely because of the fact that convergence can be
determined directly from the generators $\mathscr{L}^{(k;\alpha\beta)}$,
$\mathscr{L}^{(\alpha\beta)}$ which depend only on the coefficients of
equations (\ref{eq:hpprelim}) and (\ref{eq:hplim}); the technique does not
require us to understand precisely the way in which $U_t^{(k)}$ depends on
$k$.  The benefit of this approach is clearly demonstrated by the proof
of the singular perturbation results below, which follows as an 
application of theorem \ref{thm:qsdetk}.

\begin{rem}
Note that for the implication (a)$\Rightarrow$(b) it always suffices to
choose $\{u^{(k)}\}\subset\mathcal{D}$, as $\mathcal{D}\subset
\mathrm{Dom}(\mathscr{L}^{(k;\alpha\beta)})$.  It is only when we wish to
ensure the existence of $\{u^{(k)}\}\subset\mathcal{D}$ in the reverse
implication that we must assume $\mathcal{D}$ to be a core for all
$\mathscr{L}^{(k;\alpha\beta)}$, $k\in\mathbb{N}$.  In practice we are
almost always interested in (a)$\Rightarrow$(b), where this is irrelevant. 
\end{rem}

\begin{rem}
The proof of theorem \ref{thm:qsdetk} is easily modified to show that
$\|U_t^{(k)}\psi - U_t\psi\|\xrightarrow{k\to\infty}0$ for every
$\psi\in\mathcal{H}_0\otimes\mathcal{F}$ and $t<\infty$.  However,
uniform convergence on compact time intervals appears to be easier to 
obtain for convergence to the adjoint $U_t^*$.  As it is $U_t^*$ which 
defines the physical time evolution, we do not attempt to obtain uniform 
convergence to $U_t$.
\end{rem}

\subsection{Singular perturbations}
\label{sec:sg}

Let us start by introducing the basic assumptions.  We consider
$\mathcal{H}$ and $\mathcal{D}$ as fixed from the outset; as we will
shortly see, we must choose $\mathcal{H}_0\subset\mathcal{H}$ in a
particular way.

\begin{aspt}[Singular scaling]
\label{aspt:sing}
There exist operators $Y$, $Y^\dag$, $A$, $A^\dag$, $B$, 
$B^\dag$, $F_i$, $F_i^\dag$, $G_i$, $G_i^\dag$, $W_{ij}$, $W_{ij}^\dag$ 
with the common invariant domain $\mathcal{D}$ such that 
$$
	K^{(k)} = k^2Y+kA+B,\qquad
	L^{(k)}_i = kF_i+G_i,\qquad
	N_{ij}^{(k)} = W_{ij}
$$
for all $k\in\mathbb{N}$ and $1\le i,j\le n$.
\end{aspt}

\begin{rem}
Note that condition \ref{cond:invariance} imposes additional assumptions
which are presumed to be in force throughout.  In particular, the 
Hudson-Parthasarathy relations in condition \ref{cond:invariance} 
determine completely the form of the fourth coefficient $M_i^{(k)}$.
\end{rem}

The interpretation of this scaling is that the dynamics of the physical
system contains components that evolve on a fast time scale and on a slow
time scale.  The dynamics on the fast time scale is generated by quadratic
($\propto k^2$) term, while the coupling between the fast and slow time
scales is determined by the linear ($\propto k$) term.  As $k\to\infty$,
the dynamical evolution of the fast components of the system will thus
become increasingly singular.  However, under appropriate conditions, the
dynamical evolution of the slow components of the system will converge as
$k\to\infty$, and the time evolution of these components in the limit will
be decoupled from the fast components.  Thus, in the limit of infinite
separation of time scales, we have `adiabatically eliminated' the fast
components of the system leaving only an effective time evolution of the
slow components.  Typically the slow components of the system are the
physical quantities of interest, and the elimination of the fast
components from the model constitutes a significant simplification in many
complex physical models.

As only slow dynamics should remain in the limit, we aim to choose
$\mathcal{H}_0$ such that it contains only the slow degrees of freedom of
the system.  In order to ensure that the slow dynamics do indeed have a
well defined limit, we must impose suitable structural assumptions on the
coefficients; in particular, we must enforce that the terms which are 
linear and quadratic in $k$ do not directly generate dynamics on 
$\mathcal{H}_0$---in the latter case, the limit would surely be singular 
on $\mathcal{H}_0$.  We presently collect all the necessary structural 
assumptions; their significance can be clearly seen in the proof of our 
main result.

\begin{aspt}[Structural requirements]
\label{aspt:structural}
There is a closed subspace $\mathcal{H}_0\subset\mathcal{H}$ such that
\begin{enumerate}
\item $P_0\mathcal{D}\subset\mathcal{D}$;
\item $YP_0=0$ on $\mathcal{D}$;
\item There exist $\tilde Y$, $\tilde Y^\dag$ with the common invariant 
domain $\mathcal{D}$, so that
$\tilde YY=Y\tilde Y=P_1$;
\item $F_j^\dag P_0=0$ on $\mathcal{D}$ for all $1\le j\le n$;
\item $P_0AP_0=0$ on $\mathcal{D}$.
\end{enumerate}
Here $P_0$ denotes the orthogonal projection onto $\mathcal{H}_0$ and
$P_1$ the orthogonal projection onto $\mathcal{H}_0^\perp$. 
We choose the dense domain $\mathcal{D}_0=P_0\mathcal{D}$ in 
$\mathcal{H}_0$.
\end{aspt}

It remains to introduce the limit coefficients.  The following 
expressions emerge naturally in the proof of our main result through the 
application of theorem \ref{thm:qsdetk}.

\begin{aspt}[Limit coefficients]  
\label{aspt:coeff}
Define the operators on $\mathcal{H}_0$
$$
	K=P_0(B - A\tilde YA)P_0,\qquad
	L_i = P_0(G_i - A\tilde YF_i)P_0,
$$
$$
	M_i = -\sum_{j=1}^nP_0W_{ij}
		(G_j^\dag-F_j^\dag\tilde YA)P_0,
	\qquad
	N_{ij} = 
        \sum_{\ell=1}^n
	P_0W_{i\ell}(F_\ell^\dag\tilde YF_j+\delta_{\ell j})P_0.
$$
Then $K$, $K^\dag$, $L_i$, $L_i^\dag$, $M_i$, $M_i^\dag$, $N_{ij}$, 
$N_{ij}^\dag$ have common invariant domain $\mathcal{D}_0$.  To 
ensure that these coefficients satisfy the Hudson-Parthasarathy relations 
of condition \ref{cond:invariance}, we require
$$
	P_0(G_i - A\tilde YF_i)P_1 =
	\sum_{\ell=1}^n
        P_0W_{i\ell}(F_\ell^\dag\tilde YF_j+\delta_{\ell j})P_1 =
	\sum_{\ell=1}^n
        P_1W_{i\ell}(F_\ell^\dag\tilde YF_j+\delta_{\ell j})P_0
	= 0.
$$
\end{aspt}

Our first order of business is to verify that the coefficients $K$, $L_i$,
$M_i$, $N_{ij}$ do indeed satisfy the Hudson-Parthasarathy relations in
condition \ref{cond:invariance}.  This is proved in section
\ref{sec:singpf}.

\begin{lem}
\label{lem:unitaritysp}
Under assumptions \ref{aspt:sing}--\ref{aspt:coeff}, the coefficients $K$, 
$L_i$, $M_i$, $N_{ij}$ satisfy the Hudson-Parthasarathy relations
in condition \ref{cond:invariance}.
\end{lem}

Finally, let us state our main result, whose proof will also be given in 
section \ref{sec:singpf}.  We remind the reader that beside the above 
assumptions, conditions \ref{cond:invariance}--\ref{cond:core} are still 
presumed to be in force; in particular, one must verify separately that 
(\ref{eq:hpprelim}) and (\ref{eq:hplim}) uniquely define contraction and 
unitary cocycles and that $\mathcal{D}_0$ is a core for 
$\mathscr{L}^{(\alpha\beta)}$, $\alpha,\beta\in\mathbb{C}^n$.

\begin{thm}[Singular perturbation]
\label{thm:sing}
Under assumptions \ref{aspt:sing}--\ref{aspt:coeff}, the singularly 
perturbed equations (\ref{eq:hpprelim}) converge to the limiting equation
(\ref{eq:hplim}) on $\mathcal{H}_0$:
\begin{equation*}
    \lim_{k \to \infty} \sup_{0 \le t \le T} 
    \|U_t^{(k)*}\psi - U_t^*\psi\| 
	= 0\qquad
	\mbox{for all }\psi\in\mathcal{H}_0\otimes\mathcal{F}.
\end{equation*}   
\end{thm}

\section{Examples}
\label{sec:ex}

\subsection{Atom-cavity models with bounded coupling operators}

In this subsection, we will consider a class of physical models which 
consist of a harmonic oscillator coupled to auxiliary degrees of 
freedom.  Such models cover various applications in quantum 
optics, such as a single mode optical cavity coupled to the internal 
degrees of freedom of a collection of atoms.  Set $\mathcal{H} = 
\mathcal{H}'\otimes\ell^2(\mathbb{Z}_+)$.  On the canonical 
basis $\{\varphi_i:i\in\mathbb{Z}_+\}$ of $\ell^2(\mathbb{Z}_+)$, the 
creation, annihilation and number operators are defined as
$$
	b^\dag\,\varphi_i=\sqrt{i+1}\,\varphi_{i+1},\qquad
	b\,\varphi_i = \sqrt{i}\,\varphi_{i-1},\qquad
	b^\dag b\,\varphi_i=i\,\varphi_i.
$$
These definitions extend directly to the domain 
$\mathcal{D}'=\mathrm{span}\{\varphi_i:i\in\mathbb{Z}_+\}\subset
\ell^2(\mathbb{Z}_+)$, and we will choose the dense domain 
$\mathcal{D}=\mathcal{H}'\oten\mathcal{D}'$ in 
$\mathcal{H}$.  We now consider prelimit equations (\ref{eq:hpprelim}) 
whose coefficients have the following form on $\mathcal{D}$:
$$
	N_{ij}^{(k)} = S_{ij}^{(k)},
	\qquad L_i^{(k)} = F_{i}^{(k)}\,b^\dag + G_i^{(k)}, \qquad
	K^{(k)} = E_{11}^{(k)}\,b^\dag b + E_{10}^{(k)}\,b^\dag +
		E_{01}^{(k)}\,b + E_{00}^{(k)}.
$$
Here $E_{pq}^{(k)}$, $F_{i}^{(k)}$, $G^{(k)}_i$, $S_{ij}^{(k)}$ are 
bounded operators on $\mathcal{H}'$, and condition \ref{cond:invariance} 
is presumed to be in force (in particular, $M_i^{(k)}$ are completely 
determined by these definitions).

We begin by showing that such equations possess unique solutions which
extend to unitary cocycles on $\mathcal{H}\otimes\mathcal{F}$, and even
that $\mathcal{D}$ is a core for $\mathscr{L}^{(k;\alpha\beta)}$,
$\alpha,\beta\in\mathbb{C}$. We will establish these facts by verifying
the conditions of Fagnola, see remark \ref{rem:fagnola}.  Indeed, this is
only a minor modification of the computations performed in
\cite{Fag90,FRS94}, and we refer to those papers for a more detailed
account of the necessary steps.

\begin{lem}
The conditions of Fagnola for the existence of unique unitary cocycle 
solutions (remark \ref{rem:fagnola}) are satisfied for the class of models 
considered in this subsection.
\end{lem}

\begin{proof}
Note that all coefficients are linear combinations of operators of the 
form $X\,b^\dag b$, $X\,b^\dag$, $X\,b$, $X$, where $X$ are 
bounded operators on $\mathcal{H}'$ whose norms are bounded by a fixed 
constant $\|X\|\le K<\infty$ (for simplicity, choose $K>1$).  Arguing as 
in \cite{Fag90,FRS94}, we find
$$
	\|X(1)\cdots X(\ell)\,\psi\otimes\varphi_{p}\|\le
	\rho^\ell\sqrt{(p+\ell+m)!},
$$
where $\psi\in\mathcal{H}'$, $X(1),\ldots X(\ell)$ is an arbitrary 
selection of coefficients or their adjoints, $m$ is the number of 
occurences of $K^{(k)}$ or $K^{(k)\dag}$, and $\rho=4K^2$.
Using the elementary estimate $(\ell+m)!\le 2^{\ell+m}\ell!m!$, we can 
thus choose (see remark \ref{rem:fagnola} for notations)
$$
	c(\ell,\psi\otimes\varphi_{p}) = 
	\rho^\ell\,2^{\ell}
	\sqrt{p!\,2^p},
$$
which does indeed satisfy the necessary analyticity requirement for 
sufficiently small $\varepsilon$.  But any vector $u\in\mathcal{D}$ is the 
linear combination of a finite number of vectors of the form 
$\psi\otimes\varphi_{p}$, so that the claim follows easily.
\end{proof}

With the existence and uniqueness taken care of, it becomes
straightforward to apply our main result on singular perturbations,
theorem \ref{thm:sing}. We will demonstrate this procedure in two physical
applications: the adiabatic elimination of a single mode cavity, which 
extends the results of Gough and Van Handel \cite{GvH07}, and the 
simultaneous elimination of a single mode cavity and the excited level of 
an atom as considered in Duan and Kimble \cite{DuK04}.

\begin{exmp}[Elimination of the cavity]
We consider the case where the coupling of the oscillator with the 
external field becomes increasingly strong; in the limit, we expect to 
eliminate the oscillator entirely from the model as it will be forced 
into its ground state.  In quantum optics, this corresponds to the 
`adiabatic elimination' of an optical cavity in the strong damping limit.
To make this precise, we set
\begin{eqnarray*}
	&& S_{ij}^{(k)} = S_{ij},\quad
	F_{i}^{(k)} = kF_{i},\quad
	G_i^{(k)} = G_i,\\
	&& E_{11}^{(k)} = k^2E_{11},\quad
	E_{10}^{(k)} = kE_{10},\quad
	E_{01}^{(k)} = kE_{01},\quad
	E_{00}^{(k)} = E_{00}
\end{eqnarray*}
for $1\le i,j\le n$, where where $S_{ij}$, $F_{i}$, $G_i$, $E_{pq}$ are 
bounded operators on $\mathcal{H}'$ which are chosen such that the 
Hudson-Parthasarathy relation of condition \ref{cond:invariance} are 
satisfied for all $k\in\mathcal{N}$.  We obtain the following result,
which does not presume conditions \ref{cond:unitarity} and
\ref{cond:core}.

\begin{prop}
Suppose that $E_{11}$ has a bounded inverse. Set
$\mathcal{H}_0 = \mathcal{D}_0 = \mathcal{H}'\otimes\mathbb{C}\varphi_0$,
\begin{eqnarray*}
	&&K = E_{00} -  E_{01}(E_{11})^{-1}E_{10},\qquad
	L_i = G_i - E_{01}(E_{11})^{-1}F_i,\\
	&&N_{ij} = \sum_{\ell=1}^n S_{i\ell}\{F_\ell^*(E_{11})^{-1}F_j+
		\delta_{\ell j}\}
\end{eqnarray*}
on $\mathcal{D}_0$, and define $M_i$ through the Hudson-Parthasarathy 
condition \ref{cond:invariance}.  Then 
\begin{equation*}
    \lim_{k \to \infty} \sup_{0 \le t \le T} 
    \|U_t^{(k)*}\psi - U_t^*\psi\| 
	= 0\qquad
	\mbox{for all }\psi\in\mathcal{H}_0\otimes\mathcal{F}.
\end{equation*}   
\end{prop}

\begin{proof}
Note that condition \ref{cond:invariance} is satisfied by assumption for 
the prelimit equations (\ref{eq:hpprelim}), while condition 
\ref{cond:unitarity} is satisfied for the prelimit equations by the 
Fagnola conditions verified above.  Note that the singular scaling 
assumption \ref{aspt:sing} is also satisfied by construction.  Next, we 
turn to the structural requirements of assumption \ref{aspt:structural}.
It is immediate that $\mathcal{D}_0=P_0\mathcal{D}\subset\mathcal{D}$,
while $Y=E_{11}\,b^\dag b$ evidently vanishes on $\mathcal{D}_0$.
We define $\tilde Y$ on $\mathcal{D}$ by setting $\tilde 
Y\,\psi\otimes\varphi_0=0$ and $\tilde 
Y\,\psi\otimes\varphi_p=p^{-1}(E_{11})^{-1}\psi\otimes\varphi_p$ for
$p\ge 1$; then indeed $\tilde YY=Y\tilde Y=P_1$ on $\mathcal{D}$.
It is immediate that $F_j^\dag b$ vanishes on $\mathcal{D}_0$ and that
$A=E_{10}\,b^\dag + E_{01}\,b$ satisfies $P_0AP_0=0$.  Hence assumption 
\ref{aspt:structural} is satisfied.  Next we note that, taking into 
account the identity $X'b\tilde Yb^\dag X\,\psi\otimes\varphi_0=
X'(E_{11})^{-1}X\,\psi\otimes\varphi_0$ for any bounded operators
$X,X'$ on $\mathcal{H}'$, the coefficients defined in the proposition 
coincide precisely with the coefficients in assumption \ref{aspt:coeff},
and it is easily verified that the remaining conditions of assumption 
\ref{aspt:coeff} are also satisfied.  In particular, the limit 
coefficients satisfy the requirements of condition \ref{cond:invariance}.
But as the limit coefficients are bounded, conditions \ref{cond:unitarity} 
and \ref{cond:core} are automatically satisfied.  We have thus verified 
all the requirements of theorem \ref{thm:sing}. 
\end{proof}
\end{exmp}

\begin{exmp}[Duan-Kimble]
We  consider a three level atom coupled to a cavity which is itself 
coupled to a single output field, i.e., we set $\mathcal{H}'=\mathbb{C}^3$ 
and $n=1$.  Let us denote by $(|e\rangle, |+\rangle, |-\rangle)$ the 
canonical basis in $\mathbb{C}^3$.  In this basis we define
\begin{equation*}
	\sigma_+^{(+)} = \begin{pmatrix}0&1&0\\0&0&0\\0&0&0 \end{pmatrix},
	\qquad\qquad
	\sigma_+^{(-)} = \begin{pmatrix}0&0&1\\0&0&0\\0&0&0 \end{pmatrix}.
\end{equation*}
Define moreover
$\sigma^{(\pm)}_- = \sigma^{(\pm)*}_+$ and $P_\pm = \sigma^{(\pm)}_-
\sigma^{(\pm)}_+$. The quantum stochastic differential equation for a 
lambda system with one leg resonantly coupled to the cavity, under the 
rotating wave approximation and in the rotating frame, has the following 
coefficients:
\begin{eqnarray*}
	&&S_{11}^{(k)} = I,\qquad G_1^{(k)} = 0,\qquad
	F_1^{(k)} = k\sqrt{\gamma},\qquad
	E_{11}^{(k)} = -\frac{\gamma}{2}\,k^2, \\
	&& E_{10}^{(k)} = -E_{01}^{(k)*} = k^2g\sigma_-^{(+)},\qquad
	E_{00}^{(k)} = k(\sigma_-^{(-)}\alpha^* - 
	\sigma_+^{(-)}\alpha),
\end{eqnarray*}
where $\gamma\in\mathbb{R}_+$ and $\alpha\in\mathbb{C}$.
This model coincides with that of Duan and Kimble \cite{DuK04}, except 
that we have allowed for additional driving on the uncoupled leg of the
atom with amplitude $\alpha$.  Let us now define operators $Y,A,B,F,G,W$ 
as
\begin{equation*}
\begin{split}
  &Y = -\frac{\gamma}{2}\,b^\dag b +
  g\,(\sigma_-^{(+)}b^\dag - \sigma_+^{(+)}b), \qquad
  A = \sigma_-^{(-)}\alpha^* - \sigma_+^{(-)}\alpha, \\
  & B = 0, \qquad
  F_1 = \sqrt{\gamma}\,b^\dag,
  \qquad G_1 = 0, \qquad W_{11} = I.
\end{split}
\end{equation*}
These definitions are easily verified to satisfy condition 
\ref{cond:invariance} and assumption \ref{aspt:sing}.

We proceed to verify assumption \ref{aspt:structural}.  Set
$\mathcal{H}_0 = \mathcal{D}_0 = \mathrm{span}\{|+\rangle\otimes\varphi_0,
|-\rangle\otimes\varphi_0\}$.  Then it is easily verified that
$\mathcal{D}_0=P_0\mathcal{D}\subset\mathcal{D}$ and that 
$YP_0=F_1^\dag P_0 = P_0AP_0 = 0$ on $\mathcal{D}$.  It remains to define 
$\tilde Y$.  To this end, define the following subspaces of $\mathcal{H}$:
\begin{equation*}
  \mathcal{H}_j = \mathrm{span}\{
	|+\rangle\otimes\varphi_j,~
	|-\rangle\otimes\varphi_j,~
	|e\rangle\otimes\varphi_{j-1}\},\qquad
	j\in\mathbb{N}.
\end{equation*}
Note that $P_1\mathcal{D} = \bigoplus_{j=1}^\infty\mathcal{H}_j$ and
that the subspaces $\mathcal{H}_j$, $j\in\mathbb{N}$ are all invariant
under the action of $Y$ where,
with respect to the basis $(|+\rangle\otimes\varphi_j,
|-\rangle\otimes\varphi_j,|e\rangle\otimes\varphi_{j-1})$,
\begin{equation*}
  Y|_{\mathcal{H}_j} = 
	\begin{pmatrix} -\frac{\gamma j}{2} &0 & g\sqrt{j} \\
  0 &-\frac{\gamma j}{2} &0 \\
  -g\sqrt{j} &0 &-\frac{\gamma(j-1)}{2}
  \end{pmatrix},
\end{equation*}
We may then construct $\tilde{Y} = \bigoplus_{j=1}^{\infty}
(\tilde{Y}|_{\mathcal{H}_j})^{-1}$ on $P_1\mathcal{D}$ (and
$\tilde Y\mathcal{D}_0=0$), which reads
\begin{equation*}
  \tilde{Y}|_{\mathcal{H}_j} = -\frac{1}{d_j}\begin{pmatrix} 
	\frac{\gamma(j-1)}{2} &0 & g\sqrt{j} \\
	  0 & \frac{2d_j}{\gamma j} &0 \\
	  -g\sqrt{j} &0 &\frac{\gamma j}{2}
  \end{pmatrix},\qquad d_j =
	\frac{\gamma^2j(j-1)}{4} + g^2 j
\end{equation*}
in the same basis as above.  Thus assumption \ref{aspt:structural} has 
been verified.  We are now in a position to introduce the limit 
coefficients.  By assumption \ref{aspt:coeff}, we must evidently define
on $\mathcal{H}_0$
\begin{equation*}
   K = -\frac{|\alpha|^2\gamma}{2g^2}P_-, \qquad
   L_1 = - \frac{\alpha^*\sqrt{\gamma}}{g}
		\sigma^{(-)}_-\sigma_+^{(+)},\qquad
   N_{11} = I-2P_-.
\end{equation*}
It is again easily verified by explicit computation that the remaining 
conditions of assumption \ref{aspt:coeff} are satisfied.  Moreover, all 
the coefficients that we have introduced satisfy the Hudson-Parthasarathy 
relations of condition \ref{cond:invariance}, condition 
\ref{cond:unitarity} is satisfied for the prelimit equations by the 
Fagnola conditions, and conditions \ref{cond:invariance} and 
\ref{cond:core} are automatically satisfied as the limit coefficients are 
bounded.  Thus all the conditions of theorem \ref{thm:sing} have been 
verified.
\end{exmp}

\subsection{Coupled oscillators}

In the examples in the previous subsection, the application of theorem 
\ref{thm:sing} was significantly simplified by two convenient properties: 
the Fagnola conditions could be verified for that class of models, and the 
limit equations always had bounded coefficients.  In the present section, 
we will develop an example which enjoys neither of these properties.

We consider a single mode cavity in which one of the mirrors is allowed to
oscillate along the cavity axis (see, e.g., \cite{DJ99}).  The initial
system will therefore have two degrees of freedom: the kinematic
observables of the oscillating mirror and the usual observables associated
with the cavity mode.  This is implemented by choosing
$\mathcal{H}=L^2(\mathbb{R})\otimes \ell^2(\mathbb{Z}_+)$.  For the cavity
mode Hilbert space $\ell^2(\mathbb{Z}_+)$, we define the domain
$\mathcal{D}'$ of finite particle vectors and the creation and
annihilation operators $b^\dag$, $b$ as in the previous section.  In the 
mirror Hilbert space $L^2(\mathbb{R})$ we will choose as dense domain the
Schwartz space \cite[ch.\ 8--9]{Fol99}
$$
	\mathcal{S} = \{
		f\in C^\infty(\mathbb{R}):
		\|X^\alpha D^\beta f\|_\infty<\infty
		~ \forall\,\alpha,\beta\in\mathbb{Z}_+
	\},
$$
where we have defined the operators
$$
	Df(x) = \frac{df(x)}{dx},\qquad
	Xf(x) = x f(x)\qquad
	\forall f\in C^\infty(\mathbb{R}).
$$
Note that these operators leave $\mathcal{S}$ invariant, as do all 
operators of the following form:
$$
	(g(X)f)(x) = g(x)f(x)\qquad
	\forall f\in\mathcal{S},~
	g\in C^\infty(\mathbb{R})~\mathrm{s.t.}~
	\|D^\beta g\|_\infty<\infty~\mbox{for all}~\beta\in\mathbb{Z}_+.
$$
In addition, let us recall the following well known facts
\cite[appendix V.3]{ReS80}.  If we define the operators 
$B,B^\dag$ with invariant domain $\mathcal{S}$ and the vectors 
$\Phi_j\in\mathcal{S}$ as
$$
	B = \frac{X+D}{\sqrt{2}},\qquad
	B^\dag = \frac{X-D}{\sqrt{2}},\qquad
	\Phi_j = \pi^{-1/4}(j!)^{-1/2}(B^\dag)^je^{-x^2/2},~j\in\mathbb{Z}_+,
$$
then $\mathcal{S}'=\mathrm{span}\{\Phi_j:j\in\mathbb{Z}_+\}\subset\mathcal{S}$ 
is dense in $L^2(\mathbb{R})$, $B^\dag B\Phi_j=j\Phi_j$, 
$B\Phi_j=\sqrt{j}\,\Phi_{j-1}$, and $B^\dag\Phi_j = 
\sqrt{j+1}\,\Phi_{j+1}$.  Therefore $\mathcal{S}'$ corresponds to the 
domain of finite particle vectors in $\ell^2(\mathbb{Z}_+)$ in the usual 
isomorphism between $L^2(\mathbb{R})$ and $\ell^2(\mathbb{Z}_+)$.  In the 
following we will set $\mathcal{D}=\mathcal{S}\oten\mathcal{D}'$, but we 
will occasionally exploit $\mathcal{S}'$ whenever this 
is convenient.

We are now in a position to introduce the prelimit equations 
(\ref{eq:hpprelim}).  We set $n=1$, i.e., there is only one external 
field, and choose the following prelimit coefficients:
$$
	N_{11}^{(k)} = I,\qquad
	L_1^{(k)} = k\sqrt{\gamma}\,b^\dag,\qquad
	K^{(k)} = 
	k^2\left[i\vartheta\,(B+B^\dag)-\frac{\gamma}{2}\right]b^\dag b
	+ i\Omega\,B^\dag B,
$$
where $\vartheta,\gamma,\Omega>0$ are fixed parameters.  The physical 
interpretation of this model is that the optical frequency of the cavity 
mode is determined by the length of the cavity, and hence by the 
displacement ($B+B^\dag$) of the mirror; this accounts for the
$(B+B^\dag)b^\dag b$ term.  At the same time, we presume that the 
mirror is oscillating in a quadratic potential, which accounts for the 
$B^\dag B$ term, while the remaining terms couple the cavity mode to the 
external field.  In the strong damping limit $k\to\infty$ we expect as 
before that the cavity is eliminated, so that we obtain an effective 
interaction between the mirror and the external field.

\begin{lem}
The prelimit equations (\ref{eq:hpprelim}) each possess a unique solution 
which extends to a contraction cocycle on $\mathcal{H}\otimes\mathcal{F}$.
\end{lem}

\begin{proof}
First, note that condition \ref{cond:invariance} is satisfied for these 
equations on the domain $\mathcal{D}$.  It suffices to verify that the 
closure of $\mathscr{L}^{(k;\alpha\beta)}$ on $\mathcal{D}$ is the 
generator of a strongly continuous contraction semigroup for every 
$k\in\mathbb{N}$ and $\alpha,\beta\in\mathbb{C}$; see remark 
\ref{rem:fagnola} and \cite[theorem 1.1]{LiW06}.  Fix $k,\alpha,\beta$, 
and denote by $\mathscr{L}$ the operator $\mathscr{L}^{(k;\alpha\beta)}$ 
with the domain $\mathcal{D}$.  By an easy corollary of the Lumer-Phillips 
theorem \cite[lemma 2.1]{LiW06}, it suffices to prove that $\mathcal{D}$
is a core for $\mathscr{L}^*$.  But this follows immediately from the 
following lemma, and the proof is complete. 
\end{proof}

\begin{lem}
Let $T$ be an operator on the domain $\mathcal{D}$ of the following form:
$$
	T = c_1 + c_2 b^\dag b  
		+ c_3 B^\dag B
		+ c_4 b + c_5 b^\dag 
		+ c_6(B+B^\dag)b^\dag b,
$$
where $c_{1,\ldots,6}\in\mathbb{C}$.  Then $\mathcal{D}$ is a core for
$T^*$.
\end{lem}

\begin{proof}
As $\langle u,Tv\rangle = \langle T^\dag u,v\rangle$ for 
$u,v\in\mathcal{D}$, clearly $\mathcal{D}\subset\mbox{Dom}(T^*)$ 
where $T^*|_{\mathcal{D}}=T^\dag$.  Our goal is to show that the
closure $\mbox{cl}\,T^\dag$ of $T^\dag$ is in fact $T^*$, i.e., we must 
show that $\mbox{Dom}(T^*)\subset\mbox{Dom}(\mbox{cl}\,T^\dag)$.  
We begin by noting that
$$
	\mbox{Dom}(\mbox{cl}\,T^\dag)=\left\{
		\sum_{j,\ell=0}^\infty 
		d_{j,\ell}\,\Phi_j\otimes\varphi_\ell:
		\sum_{j,\ell=0}^\infty|d_{j,\ell}|^2<\infty~~
		\mbox{and}~~
		\sum_{j,\ell=0}^\infty|T_{j,\ell}(d)|^2<\infty\right\},
$$
where we have written
\begin{multline*}
	T_{j,\ell}(d) =
	(c_1^*+c_2^*\ell+c_3^*j)d_{j,\ell}
	+c_4^*\sqrt{\ell}\,d_{j,\ell-1} 
\\
	+c_5^*\sqrt{\ell+1}\,d_{j,\ell+1} 
	+c_6^*(\sqrt{j+1}\,d_{j+1,\ell}+\sqrt{j}\,d_{j-1,\ell})\ell.
\end{multline*}
Indeed, this space is evidently the completion of $\mathcal{S}'\oten
\mathcal{D}'$ with respect to the Sobolev norm
$\|v\|_{T^\dag}=\|v\|+\|T^\dag v\|$, and $\mathcal{D}$ is already
included in this space (this follows from the representation 
\cite[theorem V.13]{ReS80}).  It thus remains to show that 
$\mbox{Dom}(T^*)$ is included in this space as well.  To this end, let us 
define for all $N\in\mathbb{N}$
$$
	v = \sum_{j,\ell=0}^\infty 
		d_{j,\ell}\,\Phi_j\otimes\varphi_\ell,\qquad
	\sum_{j,\ell=0}^\infty|d_{j,\ell}|^2<\infty,\qquad
	u_N = \frac{\sum_{j,\ell=0}^N
		T_{j,\ell}(d)\,\Phi_j\otimes\varphi_\ell}
	{\|\sum_{j,\ell=0}^N
                T_{j,\ell}(d)\,\Phi_j\otimes\varphi_\ell\|}.
$$
Thus $\|u_N\|=1$ and $u_N\in\mathcal{D}$ for all $N\in\mathbb{N}$.
Now recall that by the Riesz representation theorem 
$v\in\mathrm{Dom}(T^*)$ 
iff $|\langle Tu,v\rangle|\le C\|u\|$ for all $u\in\mathcal{D}$.  
But if $v\not\in\mbox{Dom}(\mbox{cl}\,T^\dag)$, then 
$$
	|\langle Tu_N,v\rangle|=
	|\langle Tu_N,P_{N+1}v\rangle|=
	|\langle u_N,T^\dag P_{N+1}v\rangle|=
	\left\|\sum_{j,\ell=0}^N
                T_{j,\ell}(d)\,\Phi_j\otimes\varphi_\ell\right\|
	\xrightarrow{N\to\infty}\infty
$$
where $P_{N}$ is the orthogonal projection onto the subspace
$\mathrm{span}\{\Phi_j\otimes\varphi_\ell:j,\ell\le N\}$.
Thus $v\not\in\mbox{Dom}(\mbox{cl}\,T^\dag)$ implies
$v\not\in\mathrm{Dom}(T^*)$, so that
$\mbox{Dom}(T^*)\subset\mbox{Dom}(\mbox{cl}\,T^\dag)$ as desired.
\end{proof}

Let us now turn to the assumptions \ref{aspt:sing}--\ref{aspt:coeff}.
Clearly our coefficients are of the form required by assumption 
\ref{aspt:sing}.  Let us set $\mathcal{H}_0=L^2(\mathbb{R})\otimes
\mathbb{C}\varphi_0$ and $\mathcal{D}_0=P_0\mathcal{D}=
\mathcal{S}\oten\mathbb{C}\varphi_0$.  Then all the requirements of 
assumption \ref{aspt:structural} are clearly satisfied except that we must 
verify that
$$
	Y = 
	\left[i\vartheta\,(B+B^\dag)-\frac{\gamma}{2}\right]b^\dag b
$$
possesses an inverse $\tilde Y$ on $P_1\mathcal{D}$.  We may simply set, 
however,
$$
	\tilde Y\,\psi\otimes\varphi_0=0,\qquad
	\tilde Y\,\psi\otimes\varphi_j =
	j^{-1}
	\left[i\vartheta\,(B+B^\dag)-\frac{\gamma}{2}\right]^{-1}
	\psi\otimes\varphi_j,\quad j\ge 1.
$$
As $B+B^\dag = X\sqrt{2}$ and the function $g:\mathbb{R}\to\mathbb{C}$,
$g(x) = (i\vartheta x\sqrt{2}-\gamma/2)^{-1}$ is smooth and bounded 
together with all its derivatives, the bounded operator $\tilde Y$ leaves 
the Schwartz space invariant, and clearly $Y\tilde Y=\tilde YY=P_1$ on 
$\mathcal{D}$.  Therefore, all the conditions of assumption 
\ref{aspt:structural} are satisfied.  It should be noted that it would not 
have been sufficient to choose the smaller finite excitation domain 
$\mathcal{S}'$ in $L^2(\mathbb{R})$, despite the fact that it is left 
invariant by all prelimit coefficients, as the inverse $\tilde Y$ does not 
leave $\mathcal{S}'$ invariant.

Having verified assumptions \ref{aspt:sing} and \ref{aspt:structural}, the 
form of the limit equation (\ref{eq:hplim}) is determined by assumption 
\ref{aspt:coeff}.  In particular, we must choose the limit coefficients
$$
	N_{11} =
	\frac{i\vartheta\,(B+B^\dag)+\gamma/2}
	{i\vartheta\,(B+B^\dag)-\gamma/2},\qquad
	L_1 = 0,\qquad
	K = i\Omega\,B^\dag B
$$
on $\mathcal{D}_0$, and one can verify by straightforward manipulation 
that the remaining conditions in assumption \ref{aspt:coeff} are 
satisfied.  However, we have not yet established that the limit equation 
possesses a unique solution which extends to a unitary cocycle, and 
whether the core condition \ref{cond:core} is satisfied.  This is our next 
order of business.

\begin{lem}
The limit equation (\ref{eq:hplim}) possesses a unique solution 
which extends to a contraction cocycle on $\mathcal{H}_0\otimes\mathcal{F}$.
Moreover, the core condition \ref{cond:core} is satisfied.
\end{lem}

\begin{proof} 
First, note that condition \ref{cond:invariance} is satisfied for the
limit equation on the domain $\mathcal{D}_0$.  It suffices to verify that
the closure of $\mathscr{L}^{(\alpha\beta)}$ on $\mathcal{D}_0$ is the
generator of a strongly continuous contraction semigroup for every
$\alpha,\beta\in\mathbb{C}$: existence and uniqueness follows from
remark \ref{rem:fagnola} and \cite[theorem 1.1]{LiW06}, while the fact 
that
$\mathcal{D}_0$ is a core for the corresponding generators is then
trivially satisfied.  Note that we have
$$
	\mathscr{L}^{(\alpha\beta)} =
		\alpha^*\beta ~
	\frac{i\vartheta\,(B+B^\dag)+\gamma/2}
	{i\vartheta\,(B+B^\dag)-\gamma/2}
	-\frac{|\alpha|^2+|\beta|^2}{2}
	+ i\Omega\,B^\dag B.
$$
But as $B^\dag B$ is essentially self-adjoint on $\mathcal{S}$, the 
closure of $i\Omega\,B^\dag B$ on $\mathcal{S}$ is the generator of a 
strongly continuous unitary group (and hence of a contraction semigroup), 
while the first two terms in the expression for 
$\mathscr{L}^{(\alpha\beta)}$ constitute a bounded and dissipative 
perturbation.  The result therefore follows from a standard perturbation 
argument \cite[theorem 3.7]{Dav80}. 
\end{proof}

\begin{lem}
The solution of equation (\ref{eq:hplim}) extends to a unitary cocycle.
\end{lem}

\begin{proof}
As we have shown that there is a unique solution, it suffices to show that 
there exists a unitary process $U_t$ that satisfies equation 
(\ref{eq:hplim}).  To this end, consider the equations
$$
	S_t  = I + \int_0^t S_s\,i\Omega\,B^\dag B\,ds,\qquad
	R_t = I + \int_0^t R_s\,(S_sN_{11}S_s^\dag-I)\,d\Lambda_s.
$$
Then $S_t$ is the strongly continuous unitary group generated by the 
closure of $i\Omega B^\dag B$, while $R_t$ satisfies a Hudson-Parthasarathy
equation with bounded, though time dependent, coefficients.
The latter possesses a unitary solution as can be verified, e.g., using 
the results in \cite{Hol92}.  Moreover, using the quantum stochastic 
calculus, it is immediately verified that $U_t=R_tS_t$ satisfies equation 
(\ref{eq:hplim}), and $R_tS_t$ is clearly unitary.  
\end{proof}

We have finally verified all the conditions 
\ref{cond:invariance}--\ref{cond:core} and assumptions 
\ref{aspt:sing}--\ref{aspt:coeff}.  We can thus invoke theorem 
\ref{thm:sing}, and we find that indeed
\begin{equation*}
    \lim_{k \to \infty} \sup_{0 \le t \le T} 
    \|U_t^{(k)*}\psi - U_t^*\psi\| 
	= 0\qquad
	\mbox{for all }\psi\in\mathcal{H}_0\otimes\mathcal{F}
\end{equation*}   
with the prelimit and limit coefficients as define above.

\subsection{Finite dimensional approximations}

Suppose that we are given a quantum stochastic differential equation
of the form (\ref{eq:hplim}) with $\mathrm{dim}\,\mathcal{H}_0=\infty$.
Though such an equation may be a realistic physical model, it can not be 
simulated directly on a computer.  To perform numerical computations 
(typically one would simulate a derivative of this equation, such as a 
master equation or a quantum filtering equation), we must first 
approximate this infinite dimensional equation by one which is finite 
dimensional.  A very common way of doing this is to fix an orthonormal 
basis $\{\psi_\ell\}_{\ell\ge 0}\subset\mathcal{D}_0$, so that we can 
approximate the coefficients in the equation for $U_t$ by their 
truncations with respect to the first $k$ basis elements.  We will show 
that the solutions of the truncated equations do in fact converge to $U_t$ 
as $k\to\infty$.  Though the result is of some interest in itself, it also 
serves as an exceedingly simple demonstration of the Trotter-Kato theorem 
\ref{thm:qsdetk} which is not of the singular perturbation type (theorem
\ref{thm:sing}).

In the current setting, we will simply set $\mathcal{H}_0=\mathcal{H}$.  
We presume that the equation (\ref{eq:hplim}) is given on the domain 
$\mathcal{D}_0\oten\mathcal{E}$ and that it satisfies conditions 
\ref{cond:invariance}--\ref{cond:core}.  We also presume that 
$\mathcal{D}_0=\mathrm{span}\{\psi_\ell:\ell\in\mathbb{Z}_+\}$, where
$\{\psi_\ell\}_{\ell\ge 0}$ is an orthonormal basis of $\mathcal{H}$.
For simplicity, let us assume that $N_{ij}=\delta_{ij}$; this is not 
essential, see remark \ref{rem:trunc} below.

Define $P^{(k)}$ to be the orthogonal projection onto
$\mathrm{span}\{\psi_0,\ldots,\psi_k\}$.  We proceed to define the
equations (\ref{eq:hpprelim}) by truncating the coefficients. Since the
truncated operators will be bounded, we set $\mathcal{D}=\mathcal{H}$ and
we let $M_i^{(k)}, L_i^{(k)}$ and $K^{(k)}$ in equation
\eqref{eq:hpprelim} be given by
$$
  L_i^{(k)} = P^{(k)}L_iP^{(k)},\qquad
    M_i^{(k)} = - L_i^{(k)*},\qquad
  K^{(k)} = P^{(k)}KP^{(k)}
$$
for all $1\le i\le n$.  Note that condition \ref{cond:invariance} is thus 
satisfied, and as the coefficients are bounded condition 
\ref{cond:unitarity} is as well.  Conditions 
\ref{cond:invariance}--\ref{cond:core} are satisfied for the limit 
equation by assumption.

\begin{prop}[Finite dimensional approximation]
Under the above assumptions, the truncated equations (\ref{eq:hpprelim}) 
converge to the exact equation (\ref{eq:hplim}):
\begin{equation*}
    \lim_{k \to \infty} \sup_{0 \le t \le T} 
    \|U_t^{(k)*}\psi - U_t^*\psi\| 
	= 0\qquad
	\mbox{for all }\psi\in\mathcal{H}\otimes\mathcal{F}.
\end{equation*}   
\end{prop}

\begin{proof}
By theorem \ref{thm:qsdetk}, it suffices to show that for every
$\alpha,\beta\in\mathbb{C}^n$ and $u\in\mbox{span}\{\psi_\ell\}_{\ell \ge 
0}$, there exists a sequence $\{u^{(k)}\}\subset\mathcal{H}$ such that
$$
	u^{(k)}\xrightarrow{k\to\infty}u,\qquad
	\mathscr{L}^{(k;\alpha\beta)}u^{(k)}\xrightarrow{k\to\infty}
	\mathscr{L}^{(\alpha\beta)}u.
$$
We may simply take $u^{(k)} = P^{(k)}u$. Since $u$ is an element 
of the linear span of $\{\psi_\ell\}_{\ell\ge 0}$, there 
is a $C \in \mathbb{N}$ such that $u^{(k)} = u$ and
$\mathscr{L}^{(k;\alpha\beta)}u^{(k)}=\mathscr{L}^{(\alpha\beta)}u$ for 
all $k\ge C$. 
\end{proof}

\begin{rem}
\label{rem:trunc}
The assumption that $N_{ij}=\delta_{ij}$ is not necessary for the result 
to hold; however, care must be taken to approximate $N_{ij}$ in such a way 
that the Hudson-Parthasarathy relations of condition \ref{cond:unitarity} 
remain satisfied (simple truncation as above typically does not achieve 
this purpose).  Similarly, it is not difficult to show that the result 
still holds if only $\mathrm{span}\{\psi_\ell:\ell\in\mathbb{Z}_+\}
\subset\mathcal{D}_0$, provided that 
$\mathrm{span}\{\psi_\ell:\ell\in\mathbb{Z}_+\}$ is a core for
$\mathscr{L}^{(\alpha\beta)}$, $\alpha,\beta\in\mathbb{C}^n$.
The chief technical advantage of this simple result compared to, e.g., an 
application of the results in \cite{LiW07}, is the very strong nature of 
the convergence.
\end{rem}

\section{Proof of the main convergence theorem}
\label{sec:proofmain}

Before turning to the proof of the main convergence theorem, we provide a 
proof of the semigroup properties of $T_t^{(\alpha\beta)}$, lemma 
\ref{lem:gens}.

\begin{proof}[Proof of Lemma \ref{lem:gens}]
As $U_t$ is a contraction cocycle, it follows easily that $T_t^{(\alpha\beta)}$ 
is a strongly continuous contraction semigroup.  As
$\alpha I_{[0,t]}\in\mathfrak{S}$ for every $\alpha\in\mathbb{C}^n$ and 
$t\ge 0$ (recall that we assume that $\mathfrak{S}$ contains all simple 
functions), we obtain from (\ref{eq:hpprelim}) for $u,v\in\mathcal{D}_0$
\begin{multline*}
	\langle u\otimes e(\alpha I_{[0,t]}),
	U_t\,v\otimes e(\beta I_{[0,t]})\rangle = \\
	e^{\alpha^*\beta\,t}\left[\langle u,v\rangle +
	\int_0^t e^{-\alpha^*\beta\,s}
	\langle u\otimes e(\alpha I_{[0,s]}),
	U_s\,(\overline{\mathscr{L}}^{(\alpha\beta)}v\otimes e(\beta 
		I_{[0,s]}))\rangle
	\,ds\right],
\end{multline*}
where we have written for $u\in\mathcal{D}_0$
\begin{equation*}
	\overline{\mathscr{L}}^{(\alpha\beta)}u =
	\left(\sum_{i,j=1}^n \alpha_i^* (N_{ij}-\delta_{ij})\beta_j +
	\sum_{i=1}^n \alpha^*_i M_i +
	\sum_{i=1}^n L_i \beta_i + K
	\right)u.
\end{equation*}
Therefore, we obtain for $u,v\in\mathcal{D}_0$ using the chain rule
\begin{equation}
\label{eq:weakgen}
	\langle u,T_t^{(\alpha\beta)}v\rangle =
	\langle u,v \rangle +
	\int_0^t 
	\langle u,T_s^{(\alpha\beta)}\,\mathscr{L}^{(\alpha\beta)}v\rangle 
	\,ds,
\end{equation}
with $\mathscr{L}^{(\alpha\beta)}v$ defined as in equation 
(\ref{eq:geninvdomain}).  Using the fact that $\mathcal{D}_0$ is dense in 
$\mathcal{H}_0$ and that 
$$
	|\langle u,T_s^{(\alpha\beta)}\,\mathscr{L}^{(\alpha\beta)}v\rangle|
	\le \|u\|\,\|T_s^{(\alpha\beta)}\,\mathscr{L}^{(\alpha\beta)}v\|\le
	\|u\|\,\|\mathscr{L}^{(\alpha\beta)}v\|,
$$
we find using dominated convergence that (\ref{eq:weakgen}) holds 
identically for all $u\in\mathcal{H}_0$, $v\in\mathcal{D}_0$, and $t\ge 
0$.  Moreover, as $\|T_t^{(\alpha\beta)}\mathscr{L}^{(\alpha\beta)}v\|\le
\|\mathscr{L}^{(\alpha\beta)}v\|$ is bounded, it follows that 
$T_t^{(\alpha\beta)}\mathscr{L}^{(\alpha\beta)}v$
is Bochner integrable and thus evidently
(see, e.g., \cite[section 1.2]{Dyn65})
$$
	T_t^{(\alpha\beta)}v =
	v + 
	\int_0^t 
	T_s^{(\alpha\beta)}\,\mathscr{L}^{(\alpha\beta)}v 
	\,ds\qquad
	\forall\,t\ge 0,~v\in\mathcal{D}_0.
$$
But then we obtain using the
strong continuity of $T_t^{(\alpha\beta)}$
$$
	\lim_{t\searrow 0}\frac{T_t^{(\alpha\beta)}v-v}{t} =
	\mathscr{L}^{(\alpha\beta)}v\qquad
	\forall\, v\in\mathcal{D}_0,
$$
which establishes the claim.  The result for $T_t^{(k;\alpha\beta)}$
follows identically. 
\end{proof}

The remainder of this section is devoted to the proof of theorem
\ref{thm:qsdetk}.  We first prove weak convergence of $U_t^{(k)}$ to $U_t$
for every $t<\infty$; this is the content of the following lemma. The weak
convergence will subsequently be strengthened to strong convergence, and
ultimately to strong convergence uniformly on compact time intervals.

\begin{lem}
If for every $\alpha,\beta\in\mathbb{C}^n$ and
$u\in\mathcal{D}_0$, there exist
$u^{(k)}\in\mathrm{Dom}(\mathscr{L}^{(k;\alpha\beta)})$ so that
$$
        u^{(k)}\xrightarrow{k\to\infty}u,\qquad
        \mathscr{L}^{(k;\alpha\beta)}u^{(k)}\xrightarrow{k\to\infty}
        \mathscr{L}^{(\alpha\beta)}u,
$$
then $\langle\psi_1,U_t^{(k)}\psi_2\rangle\xrightarrow{k\to\infty}
\langle\psi_1,U_t\psi_2\rangle$ for every 
$\psi_1,\psi_2\in\mathcal{H}_0\otimes\mathcal{F}$ and $t<\infty$.
\end{lem}

\begin{proof}
As $U_t$ and $U_t^{(k)}$ are adapted, it suffices to restrict our 
attention to $\mathcal{H}_0\otimes\mathcal{F}_{t]}$.  Let $\mathcal{U}\subset
\mathcal{H}_0\otimes\mathcal{F}_{t]}$ be a total subset.
Then for any $\psi_1,\psi_2\in\mathcal{H}_0\otimes\mathcal{F}_{t]}$ and
$\varepsilon>0$, there exist $\psi_1',\psi_2'\in\mathrm{span}\,\mathcal{U}$
such that $\|\psi_1-\psi_1'\|<\varepsilon$,
$\|\psi_2-\psi_2'\|<\varepsilon$, and
\begin{multline*}
	\limsup_{k\to\infty}|\langle\psi_1,(U_t^{(k)}-U_t)\psi_2\rangle|
	\\ \le 2\,\|\psi_1-\psi_1'\|\,\|\psi_2\| +
		2\,\|\psi_2-\psi_2'\|\,\|\psi_1'\| +
	\limsup_{k\to\infty}|\langle\psi_1',(U_t^{(k)}-U_t)\psi_2'\rangle| \\
	\le 2\varepsilon\,(\|\psi_1\|+\|\psi_2\|) + 2\varepsilon^2
		+ 
	\limsup_{k\to\infty}|\langle\psi_1',(U_t^{(k)}-U_t)\psi_2'\rangle|.
\end{multline*}
It thus suffices to prove that the limit on the right vanishes for every
$\psi_1',\psi_2'\in\mathrm{span}\,\mathcal{U}$, i.e., it 
suffices to prove the result for $\psi_1,\psi_2\in\mathcal{U}$ only.

Now consider the total subset $\mathcal{U}=\{u\otimes e(f):
u\in\mathcal{D}_0,~f\in\mathfrak{S}'\}$, where $\mathfrak{S}'$ is the set 
of simple functions in $L^2([0,t];\mathbb{C}^n)$ (recall that
by assumption $\mathfrak{S}'$ is admissible, so that
$\mathcal{U}\subset\mathcal{D}_0\oten\mathcal{E}_{t]}$).
Let $\psi_i\in\mathcal{U}$, $i\in\{1,2\}$.  Then there exist 
$u_i\in\mathcal{D}_0$, $0=t_0<t_1<\cdots<t_m<t_{m+1}=t$, and 
$\alpha_{i0},\ldots,\alpha_{im}\in\mathbb{C}^n$ such that 
$\psi_i=u_i\otimes e(f_i)$ with $f_i(s)=\alpha_{ij}$ for $s\in 
[t_j,t_{j+1}[\mbox{}$.  It is not difficult to verify that, by virtue
of the cocycle property, 
$$
	\langle\psi_1,U_t^{(k)}\psi_2\rangle =
	\|e(f_1)\|\,\|e(f_2)\|\,
	\langle u_1,T_{t_1}^{(k;\alpha_{10}\alpha_{20})}
		T_{t_2-t_1}^{(k;\alpha_{11}\alpha_{21})}\cdots
		T_{t-t_m}^{(k;\alpha_{1j}\alpha_{2j})}u_2\rangle,
$$
and similarly for $\langle\psi_1,U_t\psi_2\rangle$.  In particular,
the result follows as
\begin{equation*}
\begin{split}
	&|\langle\psi_1,(U_t^{(k)}-U_t)\psi_2\rangle| \\
	&\qquad =
	\|e(f_1)\|\,\|e(f_2)\|\,|\langle u_1,
		(T_{t_1}^{(k;\alpha_{10}\alpha_{20})}\cdots
		T_{t-t_m}^{(k;\alpha_{1j}\alpha_{2j})}
		- T_{t_1}^{(\alpha_{10}\alpha_{20})}\cdots
		T_{t-t_m}^{(\alpha_{1j}\alpha_{2j})})u_2\rangle| \\
	&\qquad
	\le
	\|e(f_1)\|\,\|e(f_2)\|\,\|u_1\|\,
	\|(T_{t_1}^{(k;\alpha_{10}\alpha_{20})}\cdots
		T_{t-t_m}^{(k;\alpha_{1j}\alpha_{2j})}
		- T_{t_1}^{(\alpha_{10}\alpha_{20})}\cdots
		T_{t-t_m}^{(\alpha_{1j}\alpha_{2j})})u_2\|,
\end{split}
\end{equation*}
which converges to zero as $k\to\infty$ by the Trotter-Kato theorem.
\end{proof}

\begin{proof}[Proof of Theorem \ref{thm:qsdetk} (a)$\Rightarrow$(b)]
Weak convergence of $U_t^{(k)}$ to $U_t$ was proved in the preceding 
lemma.  But note that as $U_t^{(k)}$ are contractions and $U_t$ is
unitary, we can write
$$
	\|(U_t^{(k)}-U_t)^*\psi\|^2 =
	\langle (U_t^{(k)}-U_t)^*\psi,(U_t^{(k)}-U_t)^*\psi\rangle 
	\le 
	2\|\psi\|^2 - 2\,\mathrm{Re}(\langle
	U_t^{(k)*}\psi,U_t^*\psi\rangle).
$$
As $\langle\psi,U_t^{(k)}U_t^*\psi\rangle\to\|\psi\|^2$ as $k\to\infty$ 
follows from the previous lemma, we immediately obtain strong convergence.  
It thus remains to prove that strong convergence for every fixed time $t$ 
can be strengthened to strong convergence uniformly on compact time 
intervals.  To this end, we will appeal again to the Trotter-Kato theorem 
in a slightly different manner.

It is convenient to extend the Fock space to two-sided time,
i.e., we will consider the ampliations of all our operators to the
extended Fock space $\mathcal{\tilde F}=\Gamma_s(L^2(\mathbb{R};\mathbb{C}^n))
\cong\mathcal{F}_-\otimes\mathcal{F}$, where $\mathcal{F}_-\cong\mathcal{F}$ 
is the negative time portion of the two-sided Fock space.  We now define
the two-sided shift $\tilde\theta_t:L^2(\mathbb{R};\mathbb{C}^n)\to
L^2(\mathbb{R};\mathbb{C}^n)$ as $\tilde\theta_tf(s)=f(t+s)$, and by
$\tilde\Theta_t:\mathcal{\tilde F}\to\mathcal{\tilde F}$ its second 
quantization.  Note that $\tilde\Theta_t$ is a strongly continuous 
one-parameter unitary group, and that the cocycle property reads 
$U_{t+s}=U_s\tilde\Theta_s^*U_t\tilde\Theta_s$, etc., in terms of the 
two-sided shift.  Now define on the two-sided Fock space the operators
$$
	V_t^{(k)}=\tilde\Theta_t U_t^{(k)*},\qquad\qquad
	V_t=\tilde\Theta_tU_t^*.
$$
Then it is immediate from the cocycle property that $V_t^{(k)}$ and $V_t$ 
define strongly continuous contraction semigroups on 
$\mathcal{H}\otimes\mathcal{\tilde F}$ and 
$\mathcal{H}_0\otimes\mathcal{\tilde F}$, respectively, whose generators 
we will denote as $\mathscr{\tilde L}^{(k)}$ and $\mathscr{\tilde L}$ (see
\cite{Gre01} for a detailed study in the bounded case).  Moreover,
$$
	\|U_t^{(k)*}\psi - U_t^*\psi\| =
	\frac{\|V_t^{(k)}\,\psi_-\otimes\psi - V_t\,\psi_-\otimes\psi\|}
	{\|\psi_-\|}\qquad
	\forall\,\psi\in\mathcal{H}_0\otimes\mathcal{F},
	~\psi_-\in\mathcal{F}_-
$$
as $\tilde\Theta_t$ is an isometry, where 
$\psi_-\otimes\psi\in\mathcal{F}_-\otimes\mathcal{H}_0\otimes\mathcal{F}
\cong\mathcal{H}_0\otimes\mathcal{\tilde F}$.  Hence, by the Trotter-Kato 
theorem, it suffices to show the following: for any 
$\psi\in\mathrm{Dom}(\mathscr{L})$, there exist 
$\psi^{(k)}\in\mathrm{Dom}(\mathscr{L}^{(k)})$ such that
$\psi^{(k)}\to\psi$ and $\mathscr{L}^{(k)}\psi^{(k)}\to\mathscr{L}\psi$
as $k\to\infty$.  Let us prove this assertion.  Fix
$\psi\in\mathrm{Dom}(\mathscr{L})$ and $\lambda>0$, and define
$$
	\psi^{(k)} = \int_0^\infty e^{-\lambda t}V_t^{(k)}(\lambda\psi-
		\mathscr{L}\psi)\,dt = 
	R_\lambda^{(k)}(\lambda\psi-\mathscr{L}\psi),
$$
where $R_\lambda^{(k)}$ is the resolvent of $V_t^{(k)}$.  We have already 
shown that $\|U_t^{(k)*}\psi - U_t^*\psi\|\to 0$ as $k\to\infty$ for every 
fixed $t\ge 0$ and $\psi\in\mathcal{H}_0\otimes\mathcal{F}$, so that 
evidently
$$
	\|V_t^{(k)}\psi - V_t\psi\|\xrightarrow{k\to\infty} 0
	\qquad \forall\,\psi\in\mathcal{H}_0\otimes\mathcal{\tilde F},
	~t\ge 0.
$$
Therefore, we find by dominated convergence that
$$
	\psi^{(k)} \xrightarrow{k\to\infty}
	\int_0^\infty e^{-\lambda t}V_t(\lambda\psi-
		\mathscr{L}\psi)\,dt = 
	R_\lambda(\lambda\psi-\mathscr{L}\psi)=\psi,
$$
where we have used a standard result on resolvents \cite[section 2.1]{Dav80}.
Similarly, we find that 
$$
	\psi^{(k)}\in\mathrm{Dom}(\mathscr{L}^{(k)}),\qquad
	\mathscr{L}^{(k)}\psi^{(k)}=\mathscr{L}^{(k)}R_\lambda^{(k)}
	(\lambda\psi-\mathscr{L}\psi)=\mathscr{L}\psi+\lambda(\psi^{(k)}-\psi)
$$ 
for every $k\in\mathbb{N}$ by virtue of another standard result on 
resolvents \cite[\textit{loc.\ cit.}]{Dav80}.  Thus evidently
$\mathscr{L}^{(k)}\psi^{(k)}\to \mathscr{L}\psi$ as $k\to\infty$,
and the proof is complete.
\end{proof}

\begin{proof}[Proof of Theorem \ref{thm:qsdetk} (b)$\Rightarrow$(a)]
Choose $\alpha,\beta\in\mathbb{C}^n$, $v\in\mathcal{H}_0$,
and let $\psi=v\otimes e(\beta I_{[0,t]})$.  We obtain an estimate
on $\|(T_t^{(k;\alpha\beta)}-T_t^{(\alpha\beta)})v\|$ through the 
following steps:
\begin{multline*}
	\|(T_t^{(k;\alpha\beta)}-T_t^{(\alpha\beta)})v\| =
	\sup_{\stackrel{u\in\mathcal{H}}{\|u\|\le 1}}
	|\langle u,(T_t^{(k;\alpha\beta)}-T_t^{(\alpha\beta)})v\rangle| =
	\\ \sup_{\stackrel{u\in\mathcal{H}}{\|u\|\le 1}}
	\frac{|\langle u\otimes e(\alpha I_{[0,t]}),
	(U_t^{(k)}-U_t)\psi\rangle|}
	{e^{(|\alpha|^2+|\beta|^2)t/2}}  \le
	\sup_{\stackrel{\psi'\in\mathcal{H}\otimes\mathcal{F}}{
	\|\psi'\|\le 1}}
	\frac{|\langle\psi',(U_t^{(k)}-U_t)\psi\rangle|}
	{e^{|\beta|^2t/2}} = \\ 
	e^{-|\beta|^2t/2}
	\,\|(U_t^{(k)}-U_t)\psi\|
	\le 
	e^{-|\beta|^2t/2}
	\sqrt{\|\psi\|^2
	- 2\,\mathrm{Re}(\langle U_t^{(k)}\psi,U_t\psi\rangle)}.
\end{multline*}
But by assumption $\|(U_t^{(k)*}- U_t^*)\psi_0\|\to 0$ as $k\to\infty$
for all $\psi_0\in\mathcal{H}_0\otimes\mathcal{F}$, so that in particular
$\langle \psi,U_t^{(k)*}U_t\psi\rangle\to \|\psi\|^2$ as $k\to\infty$.
We thus obtain
$$
	\|(T_t^{(k;\alpha\beta)}-T_t^{(\alpha\beta)})v\| 
	\xrightarrow{k\to\infty}0\qquad
	\forall\,\alpha,\beta\in\mathbb{C}^n,~v\in\mathcal{H}_0,~t<\infty.
$$
Now denote the resolvents $R_\lambda^{(k;\alpha\beta)}=
(\lambda-\mathscr{L}^{(k;\alpha\beta)})^{-1}$.  We fix $\lambda>0$, 
$u\in\mathcal{D}_0$, and define
$$
	u^{(k)}=
	\int_0^\infty e^{-\lambda t}T_t^{(k;\alpha\beta)}
		(\lambda-\mathscr{L}^{(\alpha\beta)})u\,dt =
	R_\lambda^{(k;\alpha\beta)}(\lambda-\mathscr{L}^{(\alpha\beta)})u.
$$
Then it follows as in the proof of (a)$\Rightarrow$(b) that
$u^{(k)}\to u$ and $\mathscr{L}^{(k;\alpha\beta)}u^{(k)}\to
\mathscr{L}^{(\alpha\beta)}u$ as $k\to\infty$ (in particular,
$u^{(k)}\in\mathrm{Dom}(\mathscr{L}^{(k;\alpha\beta)})$).  This 
establishes the claim.

It remains to show that if $\mathcal{D}$ is a core for all
$\mathscr{L}^{(k;\alpha\beta)}$, $k\in\mathbb{N}$, then we may choose 
$\{u^{(k)}\}\subset\mathcal{D}$.  Indeed, let us fix 
$\alpha,\beta\in\mathbb{C}^n$, $u\in\mathcal{D}_0$, and construct
the sequence $u^{(k)}$ as before.  As $\mathcal{D}$ is a core for 
$\mathscr{L}^{(k;\alpha\beta)}$, we find that
$\{(v,\mathscr{L}^{(k;\alpha\beta)}v):v\in\mathcal{D}\}$ is dense
in $\mathrm{Graph}(\mathscr{L}^{(k;\alpha\beta)})$.  Therefore,
for every $k$, we can find a vector $w^{(k)}\in\mathcal{D}$ such that
$$
	\|w^{(k)}-u^{(k)}\|<k^{-1},\qquad
	\|\mathscr{L}^{(k;\alpha\beta)}w^{(k)}-
	\mathscr{L}^{(k;\alpha\beta)}u^{(k)}\|<k^{-1}.
$$
Then $w^{(k)}\to u$ and $\mathscr{L}^{(k;\alpha\beta)}w^{(k)}\to
\mathscr{L}^{(\alpha\beta)}u$, and the proof is complete.
\end{proof}

\section{Proof of the singular perturbation theorem}
\label{sec:singpf}

Before devoting ourselves to the proof of theorem \ref{thm:sing}, let us 
show that the coefficients in assumption \ref{aspt:coeff} satisfy the 
unitarity relations in condition \ref{cond:unitarity}.

\begin{proof}[Proof of Lemma \ref{lem:unitaritysp}]
We must verify the following relations:
$$
	K+K^\dag = -\sum_{i=1}^n L_iL_i^\dag,\qquad
	M_i = -\sum_{j=1}^n N_{ij}L_j^\dag,\qquad
	\sum_{j=1}^n N_{mj}N_{\ell j}^\dag =
	\sum_{j=1}^n N_{jm}^\dag N_{j\ell} =
	\delta_{m\ell}
$$
on $\mathcal{D}_0$.  The problem is that products such
as $L_iL_i^\dag$ will contain terms such as $G_iP_0G_i^\dag$, which can 
not be related directly to the unitarity conditions for $K^{(k)}$, etc.  
The additional conditions in assumption \ref{aspt:coeff} are chosen 
precisely to alleviate this problem, as they allow us to cancel the $P_0$ 
terms sandwiched inside the products.  For example, as
$P_0(G_i - A\tilde YF_i)P_1=0$,
$$
	\sum_{i=1}^n L_iL_i^\dag =
	\sum_{i=1}^n P_0(G_i - A\tilde YF_i)
		(G_i^\dag - F_i^\dag\tilde Y^\dag A^\dag)P_0.
$$
Note that, by applying condition \ref{cond:invariance} and assumption 
\ref{aspt:sing} to $L_i^{(k)}$ and $K^{(k)}$,
\begin{multline*}
	k^2(Y+Y^\dag) + k(A+A^\dag) + B+B^\dag = 
	K^{(k)}+K^{(k)\dag} = \\ -\sum_{i=1}^n L_i^{(k)}L_i^{(k)\dag} =
	-\sum_{i=1}^n (k^2F_iF_i^\dag+kF_iG_i^\dag+kG_iF_i^\dag + 
	G_iG_i^\dag),
\end{multline*}
and as this must hold for all $k$ we obtain
$$
	Y+Y^\dag = -\sum_{i=1}^n F_iF_i^\dag,\qquad
	A+A^\dag = -\sum_{i=1}^n(F_iG_i^\dag+G_iF_i^\dag),\qquad
	B+B^\dag  = -\sum_{i=1}^n G_iG_i^\dag.
$$
Note, in particular, that as $F_i^\dag P_0=0$
(assumption \ref{aspt:structural}(d))
$$
	(A+A^\dag)P_0 = -\sum_{i=1}^nF_iG_i^\dag P_0,\qquad
	P_0(A+A^\dag) = -\sum_{i=1}^nP_0G_iF_i^\dag.
$$
Thus
$$
	\sum_{i=1}^n L_iL_i^\dag =
	P_0\left[(A+A^\dag)\tilde Y^\dag A^\dag + A\tilde Y(A+A^\dag)
	-B-B^\dag-A\tilde Y(Y+Y^\dag)\tilde Y^\dag A^\dag\right]P_0.
$$
Using that $P_0AP_0=0$ (assumption \ref{aspt:structural}(e)), it follows 
that this expression coincides with $-K-K^\dag$ as required.
We now proceed to proving that the relation for $M_i$ holds.  Note that
$$
	\sum_{j=1}^n N_{ij}L_j^\dag =
	\sum_{j,\ell=1}^n P_0W_{i\ell}
		(\delta_{\ell j}+F_\ell^\dag\tilde YF_j)
	(G_j^\dag - F_j^\dag\tilde Y^\dag A^\dag)P_0,
$$
where we have used assumption (\ref{aspt:coeff}) as above.
Using again assumption (\ref{aspt:structural}) and the above relations, we 
obtain through straightforward manipulation
$$
	\sum_{j=1}^n (\delta_{\ell j}+F_\ell^\dag\tilde YF_j)
	(G_j^\dag - F_j^\dag\tilde Y^\dag A^\dag)P_0 =
	(G_\ell^\dag -F_\ell^\dag\tilde YA)P_0,
$$
so that the claim follows directly.  Next, we turn to the relations for 
$N_{ij}$.  We must establish
$$
	\sum_{j,a,b=1}^n P_0 W_{ma}
		(\delta_{aj}+F_a^\dag\tilde YF_j)
		(\delta_{bj}+F_j^\dag\tilde Y^\dag F_b)
 		W_{\ell b}^\dag P_0 = \delta_{m\ell}
$$
on $\mathcal{D}_0$, where we have used assumption (\ref{aspt:coeff}).  
Using assumption (\ref{aspt:structural}) and the above relations as 
before, we obtain through straightforward manipulation
$$
	\sum_{j=1}^n (\delta_{aj}+F_a^\dag\tilde YF_j)
	(\delta_{bj}+F_j^\dag\tilde Y^\dag F_b) =
	\delta_{ba}.
$$
Hence the unitarity relation for $W_{ij}$ establishes the claim.  
It remains to show that
$$
	\sum_{j,a,b=1}^n 
	P_0(\delta_{b\ell}+F_\ell^\dag\tilde Y^\dag F_b)
	W_{jb}^\dag 
	W_{ja}(\delta_{am}+F_a^\dag\tilde YF_m)P_0 
	= \delta_{m\ell}.
$$
Using the unitarity relation for $W_{ij}$ we find that we must show
$$
	\sum_{a=1}^n 
		P_0(\delta_{a\ell}+F_\ell^\dag\tilde Y^\dag F_a)
		(\delta_{am}+F_a^\dag\tilde YF_m)P_0 
	= \delta_{m\ell}.
$$
But this was already established above, and we are done.  
\end{proof}

We now turn to the proof of theorem \ref{thm:sing}. We will apply the
general convergence result, theorem \ref{thm:qsdetk}, which leaves the
problem of finding the appropriate sequences of vectors $u^{(k)}$.  To
this end, we employ a simple but cunning trick due to Kurtz \cite{Kur73}.

\begin{proof}[Proof of Theorem \ref{thm:sing}]
By theorem \ref{thm:qsdetk}, it suffices to show that for every
$\alpha,\beta\in\mathbb{C}^n$ and $u\in\mathcal{D}_0$, there exists a 
sequence $\{u^{(k)}\}\subset\mathcal{D}$ such that
$$
	u^{(k)}\xrightarrow{k\to\infty}u,\qquad
	\mathscr{L}^{(k;\alpha\beta)}u^{(k)}\xrightarrow{k\to\infty}
	\mathscr{L}^{(\alpha\beta)}u.
$$
Let $v\in\mathcal{D}$, and note that in the current setting 
(\ref{eq:geninvdomain}) reads
\begin{multline*}
	\mathscr{L}^{(k;\alpha\beta)}v =
	k^2\,Yv
	+ k\left(
	\sum_{i=1}^n F_i\beta_i
	- \sum_{i,j=1}^n \alpha_i^* W_{ij}F_j^\dag
	+ A
	\right)v \\
	+\left(
	\sum_{i,j=1}^n \alpha_i^*W_{ij}\beta_j
	+ \sum_{i=1}^n G_i\beta_i
	- \sum_{i,j=1}^n \alpha_i^* W_{ij}G_j^\dag
	+ B - \frac{|\alpha|^2+|\beta|^2}{2}
	\right)v \\
	:= (k^2\,Y + k\,A^{(\alpha\beta)}+B^{(\alpha\beta)})v.
\end{multline*}
We will seek $u^{(k)}$ of the form $u^{(k)}=u+k^{-1}u_1+k^{-2}u_2$ with
$u_1,u_2\in\mathcal{D}$.  Note that 
$$
	\mathscr{L}^{(k;\alpha\beta)}u^{(k)} =
	k^2\,Yu + k(Yu_1+A^{(\alpha\beta)}u) 
	+ Yu_2 + A^{(\alpha\beta)}u_1 + B^{(\alpha\beta)}u + o(1).
$$
Evidently the proof is complete if we can choose $u_1,u_2$ in such a way 
that $\mathscr{L}^{(k;\alpha\beta)}u^{(k)} = \mathscr{L}^{(\alpha\beta)}u 
+ o(1)$ for all $\alpha,\beta\in\mathbb{C}^n$ and $u\in\mathcal{D}_0$.  The 
structural requirements (assumption \ref{aspt:structural}) are chosen 
precisely so that this is the case.  First, note that $Yu=0$ by assumption 
\ref{aspt:structural}(b), as $u\in\mathcal{D}_0$.  This takes care of the 
quadratic term.  Next, the linear term ($\propto k$) must clearly vanish.  
Hence it must be the case that
$$
	\mathrm{span}\{Au,W_{ij}F_j^\dag u,F_iu:
	u\in\mathcal{D}_0,~i,j=1,\ldots,n\} \subset
	\{Yv:v\in\mathcal{D}\}.
$$
But note that $W_{ij}F_j^\dag u=0$ for all 
$u\in\mathcal{D}_0$ by assumption \ref{aspt:structural}(d), while 
assumptions \ref{aspt:structural}(d,e) imply that $F_i=P_1F_i$ and
$AP_0=P_1AP_0$.  Thus $F_iv=Y\tilde YF_iv$ and $Av=Y\tilde YAv$ for all 
$v\in\mathcal{D}_0$ by assumption \ref{aspt:structural}(c), which 
establishes the claim.  In particular, we will choose
$u_1 = -\tilde YA^{(\alpha\beta)}u$, which cancels the linear term.
It remains to choose $u_2\in\mathcal{D}$ so that
$$
	\lim_{k\to\infty}\mathscr{L}^{(k;\alpha\beta)}u^{(k)} =
	Yu_2 
	+ (B^{(\alpha\beta)} - A^{(\alpha\beta)}\tilde YA^{(\alpha\beta)})u
	= \mathscr{L}^{(\alpha\beta)}u.
$$
Note that even if we did not know in advance the form of 
$K,L_i,M_i,N_{ij}$, we must require that
$\mathscr{L}^{(\alpha\beta)}u\in\mathcal{H}_0$ for every 
$u\in\mathcal{D}_0$ in order that the limit 
equation leaves $\mathcal{H}_0$ invariant.  Hence, if the limit of the 
slow dynamics exists, we must be able to choose $u_2$ so that
$$
	P_1Yu_2 = -P_1(B^{(\alpha\beta)} - A^{(\alpha\beta)}\tilde 
		YA^{(\alpha\beta)})u.
$$
But $P_1Y=Y\tilde YY=YP_1=Y$ by assumptions \ref{aspt:structural}(b,c),
so that evidently we must have
\begin{multline*}
	\mathscr{L}^{(\alpha\beta)}v = 
	P_0(B^{(\alpha\beta)} - A^{(\alpha\beta)}\tilde YA^{(\alpha\beta)})v
	=
	\\
	P_0\left\{
	\sum_{i,j=1}^n \alpha_i^* \left(
	\sum_{\ell=1}^nW_{i\ell}(F_\ell^\dag\tilde YF_j+\delta_{\ell j})
	+\sum_{\ell=1}^n F_j\tilde YW_{i\ell}F_\ell^\dag
	\right)
	\beta_j\right.
	\\
	- \sum_{i,j=1}^n \alpha_i^*
		\left(W_{ij}(G_j^\dag-F_j^\dag\tilde YA)
			-A\tilde YW_{ij}F_j^\dag\right)
	\\
	+ \sum_{i=1}^n \left(
	G_i - A\tilde YF_i - F_i\tilde YA
	\right)\beta_i
	+ B - A\tilde YA
	- \frac{|\alpha|^2+|\beta|^2}{2}
	\\
	\left.
	- \left[\sum_{i=1}^n F_i\beta_i\right]
	\tilde Y
	\left[\sum_{i=1}^n F_i\beta_i\right]
	- \left[\sum_{i,j=1}^n \alpha_i^*W_{ij}F_j^\dag\right]
	\tilde Y
	\left[\sum_{i,j=1}^n \alpha_i^* W_{ij}F_j^\dag\right]
	\right\}v
\end{multline*}
for all $v\in\mathcal{D}_0$, $\alpha,\beta\in\mathbb{C}^n$.  It remains to 
verify two things: that this expression does indeed define, for every
$\alpha,\beta\in\mathbb{C}^n$, the generator of the semigroup 
$T_t^{(\alpha\beta)}$, and that $Y$ is invertible on the range of
$P_1(B^{(\alpha\beta)} - A^{(\alpha\beta)}\tilde YA^{(\alpha\beta)})P_0$.
But the latter follows immediately from assumption 
\ref{aspt:structural}(c), so it remains to deal with the former.

In order for the above expression to define the generator of 
$T_t^{(\alpha\beta)}$, at the very least the last two terms must 
vanish---after all, they are quadratic in $\beta_i$ and $\alpha_i^*$, 
respectively, which is inconsistent with (\ref{eq:geninvdomain}).
However, this is immediate from assumption \ref{aspt:structural}(d), as
it implies that $F_j^\dag P_0 = P_0F_i=0$.  Simplifying further, we find 
that for all $v\in\mathcal{D}_0$
\begin{multline*}
	\mathscr{L}^{(\alpha\beta)}v = 
	P_0\left\{
	\sum_{i,j,\ell=1}^n \alpha_i^*W_{i\ell}(F_\ell^\dag\tilde YF_j+\delta_{\ell j})
	\beta_j
	- \sum_{i,j=1}^n \alpha_i^* 
		W_{ij}(G_j^\dag-F_j^\dag\tilde YA)
	\right.
	\\
	\left.
	+ \sum_{i=1}^n (G_i - A\tilde YF_i)\beta_i
	+ B - A\tilde YA
	- \frac{|\alpha|^2+|\beta|^2}{2}
	\right\}v.
\end{multline*}
Thus evidently the coefficients of assumption \ref{aspt:coeff} emerge 
naturally from our approach (by lemma \ref{lem:gens}), and the proof is 
complete by lemma \ref{lem:unitaritysp} and conditions 
\ref{cond:invariance}--\ref{cond:core}.
\end{proof}

\bibliographystyle{plain}
\bibliography{ref}

\begin{thebibliography}{10}

\bibitem{AFLu90}
L.~Accardi, A.~Frigerio, and Y.~G. Lu.
\newblock The weak coupling limit as a quantum functional central limit.
\newblock {\em Commun. Math. Phys.}, 131:537--570, 1990.

\bibitem{Bar03}
A.~Barchielli.
\newblock Continual measurements in quantum mechanics and quantum stochastic
  calculus.
\newblock In S.~Attal, A.~Joye, and C.-A. Pillet, editors, {\em Open Quantum
  Systems III: Recent Developments}, pages 207--292. Springer, 2006.

\bibitem{BoS07}
L.~Bouten and A.~Silberfarb.
\newblock Adiabatic elimination in quantum stochastic models, 2007.
\newblock {C}ommun.~Math.~Phys., at press.

\bibitem{BrR87}
O.~Bratteli and D.W. Robinson.
\newblock {\em Operator algebras and quantum statistical mechanics 1}.
\newblock Springer-Verlag, Berlin Heidelberg, second edition, 1987.

\bibitem{Dav80}
E.~B. Davies.
\newblock {\em One-parameter semigroups}.
\newblock Academic Press Inc (London) Ltd, 1980.

\bibitem{DDR06}
J.~Derezinski and W.~{De Roeck}.
\newblock Extended weak coupling limit for {Pauli-Fierz} operators, 2006.
\newblock {C}ommun.~Math.~Phys., at press.

\bibitem{DJ99}
A.~C. Doherty and K.~Jacobs.
\newblock Feedback-control of quantum systems using continuous
  state-estimation.
\newblock {\em Phys. Rev. A}, 60:2700, 1999.

\bibitem{DuK04}
L.-M. Duan and H.~J. Kimble.
\newblock Scalable photonic quantum computation through cavity-assisted
  interaction.
\newblock {\em Phys. Rev. Lett.}, 92:127902, 2004.

\bibitem{Dyn65}
E.~B. Dynkin.
\newblock {\em Markov processes {I}}, volume 122 of {\em Grundlehren der
  Mathematischen Wissenschaften}.
\newblock Academic Press, New York, 1965.

\bibitem{Fag90}
F.~Fagnola.
\newblock On quantum stochastic differential equations with unbounded
  coefficients.
\newblock {\em Probab. Th. Rel. Fields}, 86:501--516, 1990.

\bibitem{Fag93}
F.~Fagnola.
\newblock Characterization of isometric and unitary weakly differentiable
  cocycles in {Fock} space.
\newblock In L.~Accardi, editor, {\em Quantum Probability and Related Topics},
  volume VIII, pages 143--164. World Scientific, 1993.

\bibitem{FRS94}
F.~Fagnola, R.~Rebolledo, and C.~Saavedra.
\newblock Quantum flows associated to master equations in quantum optics.
\newblock {\em J. Math. Phys.}, 35:1--12, 1994.

\bibitem{Fol99}
G.~B. Folland.
\newblock {\em Real analysis}.
\newblock John Wiley \& Sons Inc., New York, second edition, 1999.
\newblock Modern techniques and their applications, A Wiley-Interscience
  Publication.

\bibitem{Gou05}
J.~Gough.
\newblock Quantum flows as {Markovian} limit of emission, absorption and
  scattering interactions.
\newblock {\em Commun. Math. Phys.}, 254:489--512, 2005.

\bibitem{GvH07}
J.~Gough and R.~{van Handel}.
\newblock Singular perturbation of quantum stochastic differential equations
  with coupling through an oscillator mode.
\newblock {\em J. Stat. Phys.}, 127:575, 2007.

\bibitem{Gre01}
M.~Gregoratti.
\newblock The {Hamiltonian} operator associated with some quantum stochastic
  evolutions.
\newblock {\em Commun. Math. Phys.}, 222:181--200, 2001.

\bibitem{Hol92}
A.~S. Holevo.
\newblock Time-ordered exponentials in quantum stochastic calculus.
\newblock In {\em Quantum probability \& related topics}, QP-PQ, VII, pages
  175--202. World Sci. Publ., River Edge, NJ, 1992.

\bibitem{HuP84}
R.~L. Hudson and K.~R. Parthasarathy.
\newblock Quantum {It\^o's} formula and stochastic evolutions.
\newblock {\em Commun. Math. Phys.}, 93:301--323, 1984.

\bibitem{Kur73}
T.~G. Kurtz.
\newblock A limit theorem for perturbed operator semigroups with applications
  to random evolutions.
\newblock {\em J. Funct. Anal.}, 12:55--67, 1973.

\bibitem{LiW06}
J.~M. Lindsay and S.~J. Wills.
\newblock Construction of some quantum stochastic operator cocycles by the
  semigroup method.
\newblock {\em Proc. Indian Acad. Sci. (Math. Sci.)}, 116:519--529, 2006.

\bibitem{LiW07}
J.~M. Lindsay and S.~J. Wills.
\newblock Quantum stochastic operator cocycles via associated semigroups.
\newblock {\em Math. Proc. Camb. Phil. Soc.}, 142:535--556, 2007.

\bibitem{Par92}
K.~R. Parthasarathy.
\newblock {\em An Introduction to Quantum Stochastic Calculus}.
\newblock Birkh\"auser, Basel, 1992.

\bibitem{ReS80}
M.~Reed and B.~Simon.
\newblock {\em Functional Analysis}, volume~1 of {\em Methods of Modern
  Mathematical Physics}.
\newblock Elsevier, 1980.

\end{thebibliography}

\end{document}